\documentclass[11pt,a4paper,reqno]{amsproc}
\usepackage{color}

\input{header2.sty}

\begin{document}

\title[Perturbations of random Schr\"odinger operators]{Perturbations of 
	continuum random Schr\"odinger operators\\ 
	with applications to Anderson orthogonality \\
	and the spectral shift function
	}

\author[A.\ Dietlein]{Adrian Dietlein}
\address[A.\ Dietlein]{Institute of Science and Technology Austria,
  Am Campus 1,
  3400 Klosterneuburg, Austria}
\email{adrian.dietlein@ist.ac.at}

\author[M.\ Gebert]{Martin Gebert}
\address[M.\ Gebert]{Department of Mathematics, University of California, Davis, Davis, CA 95616, USA}
\thanks{M.G. was  supported by the DFG under grant GE 2871/1-1.}
\email{mgebert@math.ucdavis.edu}

\author[P.\ M\"uller]{Peter M\"uller}
\address[P.\ M\"uller]{Mathematisches Institut,
  Ludwig-Maximilians-Universit\"at M\"unchen,
  Theresienstra\ss{e} 39,
  80333 M\"unchen, Germany}
\email{mueller@lmu.de}

\begin{abstract}
	We study effects of a bounded and compactly supported perturbation on multi-dimensional continuum random 
	Schr\"odinger operators in the region of complete localisation. Our main emphasis is on Anderson orthogonality 
	for random Schr\"odinger operators. Among others, we prove that Anderson orthogonality does occur for Fermi 
	energies in the region of complete localisation with a non-zero probability.  
	This partially confirms recent non-rigorous findings
	[V.\ Khemani et al., Nature Phys. \textbf{11}, 560--565 (2015)].
	The spectral shift function plays an important role in our analysis of Anderson orthogonality. We identify it
	with the index of the corresponding pair of spectral projections and explore the consequences thereof. 	
	All our results rely on the main technical estimate of this paper which guarantees
	separate exponential decay of the disorder-averaged Schatten $p$-norm
	of $\chi_{a}(f(H) - f(H^{\tau})) \chi_{b}$ in $a$ and $b$. Here, $H^{\tau}$ is a perturbation of the random 
	Schr\"odinger operator $H$, $\chi_{a}$ is the multiplication operator corresponding
	to the indicator function of a unit cube centred about $a\in\R^{d}$, and
	$f$ is in a suitable class of functions of bounded variation with
	distributional derivative supported in the region of complete localisation
	for $H$.
\end{abstract}

\maketitle

\tableofcontents


\section{Introduction}

In this paper we consider an alloy-type random Schr\"odinger operator $H$ in multi-dimensional Euclidean space $\R^{d}$ and its perturbation $H^{\tau}$ by a bounded and compactly supported potential $W$ with coupling constant $\tau$.
Our main goal is to investigate Anderson orthogonality for random Schr\"odinger operators in the region of complete localisation. These results are compiled in Section~\ref{subsec:aoc}. Anderson orthogonality refers to the vanishing of the ground-state overlap \eqref{eq:DefOverlap} of the unperturbed and the perturbed system in the thermodynamic limit and constitutes a fundamental property of fermionic systems \cite{artAOC1967And, RevModPhys.62.929, mahan-book00}. It was believed in physics for many years that this phenomenon was absent for localised systems \cite{PhysRevB.65.081106}. However, the quest for an understanding of many-body localisation has led to a revival of such questions in physics, and 
the very recent non-rigorous studies \cite{Khemani-NaturePhys15, PhysRevB.92.220201}
revealed interesting phenomena arising in this situation. They suggest: \; (i)~\,~Anderson orthogonality \emph{does} occur for random Schr\"odinger operators when the Fermi energy is in the region of complete localisation, but with a probability strictly between 0 and 1.  \; (ii)~\,~If the overlap vanishes in the thermodynamic limit, it vanishes even exponentially fast in the system size (named ``statistical'' Anderson orthogonality). 
In this paper, we provide the first partial mathematical understanding of these findings: 
for a Fermi energy $E$ in the region of complete localisation, we prove in Theorem~\ref{th:IntMainRes} 
convergence of the expectations and almost-sure convergence of the finite-volume overlap to a random variable $S(E)$, which is given by a Fredholm determinant related to the difference of Fermi projections $\id_{(-\infty,E]}(H) - \id_{(-\infty,E]}(H^{\tau})$.
Anderson orthogonality occurs if and only if this difference has an eigenvalue of modulus 1. If $S(E) \neq 0$, we prove exponential speed of convergence in the thermodynamic limit. We also identify conditions where both presence and absence of Anderson orthogonality occur with positive probability. 

As an aside we obtain new criteria which offer a possibility to detect spectral phases of random Schr\"odinger operators other than the region of complete localisation.  They are  formulated in Corollaries~\ref{phase-transition} and ~\ref{rm:notFMB}.

A main tool in our analysis of Anderson orthogonality is the spectral shift function of the 
pair $(H, H^{\tau})$ for energies in the region of complete localisation. We obtain several new results on this subject which are compiled in Section~\ref{subsec:ssf}. In Theorem~\ref{th:SSFandTransOp} we establish convergence of the finite-volume spectral shift functions towards the infinite-volume spectral shift function and identify the latter with the index of the corresponding pair of spectral projections. We show H\"older continuity of the averaged spectral shift function in Lemma~\ref{lem:HoelderCont} and find a \emph{lower} bound in Theorem~\ref{lem:SSFcont} which establishes strict positivity of the averaged spectral shift function under appropriate conditions. 

All of the aforementioned results are applications of the main technical estimates of this paper, which 
are summarised in Section~\ref{sect:TraceClassEst} and are of their own interest.
Theorem~\ref{th:TrClBounds1} states that the probabilistic expectation of fractional moments of the Schatten $p$-norm of $\chi_{a}\big(f(H_G) - f(H_G^{\tau})\big) \chi_{b}$ exhibits \emph{separate} exponential decay in $a$ and $b$. This means that we have decay in both the distance of $a$ to the support of $W$ and the distance of $b$ to the support of $W$. Here, $H_G^{(\tau)}$ is the Dirichlet restriction of $H^{(\tau)}$ to an open region $G \subseteq \R^{d}$,
$\chi_{a}$ is the multiplication operator corresponding to the indicator function of a unit cube centred about $a\in\R^{d}$ and $f$ is a function of bounded variation 
on $\R$ whose support is bounded from the right and which is constant except for energies in a compact 
interval inside the region of complete localisation for $H$. (A non-trivial example in this class of functions is the indicator function $f=\id_{(-\infty,E]}$, where $E$ is in the region of complete localisation.)
An analogous estimate holds for perturbations by boundary conditions and is stated in Theorems~\ref{th:TrClBounds3}. Taken together, they allow to control macroscopic limits of differences 
$f(H_{L}) - f(H^{\tau}_L)$ in Schatten $p$-norms for finite-volume operators restricted to 
boxes of size $L$, see Theorem \ref{Thm:prod}.

Finally, we advertise two points in our proofs whose validity is not restricted to the region of complete localisation and which may be of independent interest. 
Firstly, in order to express $f(K)$, $K$ a self-adjoint operator, in terms of the resolvent we formulate a 
suitably adapted version of the Helffer--Sj\"ostrand formula in Lemma~\ref{lem:LocHelfSjo} which applies to compactly supported functions of bounded variation. In order to treat such general functions $f$ we trade regularity of $f$ against boundedness of (spatially localised) resolvents. 
The second result is a scale-free unique continuation principle for averaged local traces of non-negative functions of \emph{infinite-volume} ergodic random Schr\"odinger operators in Lemma~\ref{lem:UCP}. It follows from a scale-free unique continuation principle for spectral projections of \emph{finite-volume} random Schr\"odinger operators \cite{MR3106507, Nakic:2015is} by a suitable limit procedure. 
We are not aware that this consequence has been noticed elsewhere.


\section{The model}

We consider a random Schr\"odinger operator \cite{carlac1990random,pastfig1992random, MR1935594, AizWarBook}
\begin{equation}
	\label{eq:TheOperator} 
	\omega \mapsto H_{\omega} := H_0+V_\omega := H_0 +\lambda \sum_{k\in\Zd} \omega_k u_k
\end{equation}
acting on a dense domain in the Hilbert space $L^{2}(\Rn)$ for $d\in\N$ and $\lambda\geq 0$. Here, 
$H_0$ is a non-random self-adjoint operator and $\omega \mapsto V_{\omega}$ is a random alloy-type potential subject 
to the following assumptions.

\begin{MyDescription}
\item[(K)]{
	\label{assK}
	The non-random operator $H_0 := -\Delta+V_0$ consists of the non-negative Laplacian $-\Delta$ on 
	$d$-dimensional Euclidean space $\Rn$ and of a deterministic, $\Zd$-periodic and bounded background potential
	$V_0 \in L^{\infty}(\R^{d})$.  
	}
\item[(V1)]{ 
	\label{assV1}	
	The family of canonically realised random coupling constants $\omega:=(\omega_k)_{k\in\Zd} 
	\in \R^{\Z^{d}}$ is distributed independently and identically according to the Borel probability measure 
	$\PP := \bigotimes_{\Z^{d}}P_{0}$ on $\R^{\Z^{d}}$. We write $\E$ for the corresponding expectation. 
	The single-site distribution $P_0$ is absolutely continuous with respect to Lebesgue measure on $\R$ 
	and the corresponding Lebesgue density $\rho$ obeys $0 \le \rho\in L^\infty_{c}(\R)$ and 
	$\supp(\rho) \subseteq [0,1]$.
  }
\item[(V2)]{
	\label{assV2}	
	The single-site potentials $u_k(\,\cdot\,):=u(\,\cdot\,-k)$, $k\in\Zd$, are translates 
	of a non-negative function $0 \le u\in L^\infty_{c}(\R^d)$. Moreover, there exists a constant
	$C_{u,-}$ such that the covering condition
	\begin{equation}
		\label{uk-bounds}
		0< C_{u,-} \leq \sum_{k\in\Zd} u_k  
	\end{equation}
	holds.
	}
\end{MyDescription}

\noindent
The assumption $\supp(\rho) \subseteq [0,1]$ on the support of the single-site probability density is 
made for convenience only and is not stronger than requiring $\supp(\rho)$ to be compact. 
The regularity of $P_0$ as well as the covering condition \eqref{uk-bounds}  
guarantee that the fractional-moment bounds \eqref{eq:DefFMB} and the a priori bounds
\eqref{eq:AppAPriori} below hold in fair generality for our model, see also Remarks~\ref{rem:FMBS} below. 
	
Throughout the paper we drop the subscript $\omega$ from $H$ and other quantities when we think of 
these quantities as random variables (as opposed to their particular realisations). We write 
\textit{almost surely} if an event occurs  $\mathbb{P}$-almost surely, and, given an $E$-dependent statement for
$E \in B \in \Borel(\R)$,  we write \textit{for a.e.\ $E\in B$}, if this statement holds for 
Lebesgue-almost every $E \in B$. 

The alloy-type random Schr\"odinger operator \eqref{eq:TheOperator} is $\Zd$-ergodic with respect to 
lattice translations. It follows that there exists a closed set $\Sigma\subset\R$, the non-random spectrum 
of $H$, such that $\Sigma=\sigma(H)$ holds almost surely \cite{pastfig1992random}. Here, $\sigma(H)$ is 
the spectrum of $H$.  

Given an open subset $G\subseteq \Rn$, we write $H_G$ for the Dirichlet restriction of $H$ to $G$. 
We define the random finite-volume eigenvalue counting function
\beq
	\label{eq:ev-count}
	\R \ni E \mapsto N_L(E):= \Tr \left(\id_{(-\infty,E]}(H_L)\right)
\eeq
for $L>0$, where $\id_{B}$ stands for the indicator function of a set $B\in\Borel(\R)$, $H_L:=H_{\Lambda_L}$ 
and $\Lambda_L:=(-L/2,L/2)^{d}$ for the open cube about the origin of side-length $L$.
We write $|\Lambda_{L}| = L^{d}$ for the Lebesgue volume of the latter. 
The regularity assumption \ref{assV1} on the single-site distribution ensures a (standard) Wegner 
estimate \cite{MR2362242,Kirsch:2007il,MR2378428}. The Wegner estimate and ergodicity imply that the limit
\beq
	\label{eq:DOSFinVol}
	\cN(E):= \lim_{L\to\infty} \frac 1 {|\Lambda_L|} N_L(E)
\eeq
exists for all $E\in\R$ almost surely. The limit function $\cN$ is called the 
\emph{integrated density of states} of $H$. 
Moreover, the Wegner estimate implies that it is an absolutely continuous 
function under our assumptions. Its Lebesgue derivative $\cN'$ is called the \emph{density of states}.

This paper is concerned with local perturbations of the random Schr\"odinger operator $H$.
For a bounded and compactly supported function $W \in L^{\infty}_c(\Rn)$, not necessarily of definite sign, we set 
\beq
	\label{def:perturbedoperator}
	H^{\tau}:=H+\tau W,
\eeq
where $\tau\in [0,1]$ is the tunable strength of the perturbation and $H^0=H$ is the unperturbed operator. 
We assume, from now on, that $W$ is fixed. The dependence of our results on the strength of the perturbation 
will be analysed through the dependence on $\tau$. 
We note that $H^{\tau}$ has the same integrated density of states as $H$.

All results in this paper refer to energies in the \emph{region of complete localisation} -- following the terminology of \cite{Germinet2006}. We characterise it in terms of fractional-moment bounds 
\cite{artRSO2006AizEtAl2}, see also \cite{MR2303305}. To this end we need to introduce some more notation. 
Given a self-adjoint operator $A$, we denote by $R_z(A):=(A-z)^{-1}$ its resolvent at 
$z\in \C\setminus \sigma(A)$. We write $Q_a:=a+[-1/2,1/2]^d$ for the (closed) unit cube centred at 
$a\in\R^d$ and $\chi_a:=\id_{Q_a}$ for its indicator function.
We denote the supremum norm on $\R^d$ by $|x|:= \sup_{j=1,\ldots,d} |x_j|$, where 
$x= (x_{1}, \ldots, x_{d}) \in\Rn$.
Finally, $\dist(U,V):=\inf \{  |x-y|:\ x\in U,\ y\in V\}$ stands for the distance of two subsets 
$U,V\subset \R^d$ with respect to the supremum norm.

\begin{definition}[Fractional-moment bounds]
	\label{DefFMB}
	We write $E\in \FMB:=\FMB(H)$ if there exists a neighbourhood $U:=U_E$ of $E$ in $\R$ and an exponent 
	$0<s<1$	such that the following holds: there exist constants $C,\mu>0$ such that for all open subsets 
	$G\subseteq \Rn$ and all $a,b \in \Rn$ the bound
	\begin{equation}
		\label{eq:DefFMB}
		\sup_{\substack{E'\in U_E\\ \eta\neq 0}} \mathbb{E}\left[\|\chi_a R_{E'+\i\eta}(H_G) \chi_b\|^s \right] 
		\leq C \e^{-\mu|a-b|}
	\end{equation}
	holds true. Here, $\|\boldsymbol\cdot\|$ stands for the operator norm.
\end{definition}

\begin{remarks}
\label{rem:FMBS}
\item 
	\label{FMBalways}
	Suppose that \eqref{eq:DefFMB} holds for \emph{some} exponent
	$0<s<1$, then it holds for \emph{all} exponents $0<s<1$ with constants $C$ and $\mu$ depending on $s$. 
	The proof  for discrete models in \cite[Lemma B.2]{artRSO2001AizEtAl} extends to 
	the continuum setting. 
\item
	\label{FMBalways2}
	We show in Lemma~\ref{Lem:PertPersPot} that the stability of the regime of complete localisation 
	$\FMB(H)=\FMB(H^\tau)$ holds for perturbations 
	$\tau W$ as considered in this paper. Hence, even though our proofs require mostly fractional moment bounds 
	for both $H$ and $H^\tau$, it suffices to postulate \eqref{eq:DefFMB} for the unperturbed operator 
	$H$ only and -- according to the previous remark -- with a fixed exponent $s$. 
\item 
	\label{rem:aPrioriholds}
	The proof of \eqref{eq:DefFMB} in \cite{artRSO2006AizEtAl2} for alloy-type random Schr\"odinger operators relies 
	-- as in the discrete case -- on a priori estimates for fractional moments of the resolvent. 
	In fact, the following holds 
	for our model: for every $0<s<1$ and 
	every bounded interval $I \subset \R$ there exists a finite constant $C_{s}>0$ such that
	\begin{equation}
		\label{eq:AppAPriori}
		\sup_{\tau\in[0,1]}\sup_{\substack{E\in I,\eta\neq 0\\G\subseteq\Rn,\, a,b\in\Rn}} \mathbb{E} 
		\left[ \Vert \chi_{a}R_{E+\i\eta}(H^{\tau}_G)\chi_{b} \Vert^{s} \right] \leq C_{s}.
	\end{equation}
	Except for the supremum in $\tau$, this is the content of Lemma 3.3 and the subsequent Remark (4) in \cite{artRSO2006AizEtAl2}. The relevant $\tau$-dependence of the constant in that lemma enters in 
	(3.31) upon replacing $z$ by $z - \tau W$. This leads to a polynomial $\tau$-dependence of the following norm estimates
	and implies uniformity in	$\tau\in[0,1]$. 
\end{remarks}


\section{Results}
\label{sect:results}

The results in Subsection \ref{sect:TraceClassEst} are at the heart of the proofs for the applications on the 
spectral shift function and Anderson's orthogonality in the following two subsections. 

\subsection{Schatten-class estimates}
\label{sect:TraceClassEst}

First, we fix our notation. Given $p>0$ and a compact operator $A$ on a separable Hilbert space, 
we say that $A\in \S^p$, the \emph{Schatten $p$-class}, if 
$	\|A\|_p := \l(\Tr (|A|^p)\r)^{1/p}$,  
where $\Tr(\pmb\cdot)$ stands for the trace. 
This induces a complete norm on the linear space $\mathcal{S}^p$ for $p\geq 1$,
while for $0<p<1$ it induces a complete quasi-norm for which the adapted triangle inequality 
\beq\label{lem:conkav}
	\|A+B\|_p^p \leq \|A\|_p^p+\|B\|_p^p
\eeq
holds \cite[Thm. 2.8]{McCarthy}.
We state our results for right-continuous functions $f$ on $\R$ which are of bounded variation
\beq
\TV(\R):=\big\{ f:\R\to\R\ \text{measurable, right-continuous}, \ \NTV(f)<\infty\big\}.
\eeq
Here, the \emph{total variation} is given by $\NTV(f):= \sup_{(x_{p})_{p} \in \mathcal P}\sum_{p} 
|f(x_{p+1}) - f(x_p)|$ with the supremum being taken over the set $\mathcal P$ of all finite partitions of $\R$. 
We write $\TV_c(\R)$ for the subspace of all functions in $\TV(\R)$ with compact support. 
Finally, given a bounded interval $I=[I_-,I_+]\subset\R$, we define the vector space of functions 
\begin{equation}
	\mathcal{F}_I := \Big\{f\in \TV(\R):\  f\big|_{(-\infty,I_-]}\equiv \text{ const.} , \, 
	f\big|_{[I_+,\infty)}\equiv 0 \, \Big\} \subset L^{\infty}(\R),
\end{equation}
which is a normed space with respect to the total variation.

\begin{theorem}
	\label{th:TrClBounds1}
	Fix $p>0$, $0<s<1$ and let $I \subset \FMB$ be a compact interval. Then there exist finite constants 
	$C,\mu>0$ such that the following holds: for all $G\subseteq\Rn$ open, $a,b\in \R^d$, $\tau \in [0,1]$ 
	and $f \in \mathcal{F}_I$ we have 
	\begin{equation}
		\label{eq:TrClBounds1Stat}
		\mathbb{E} \big[ \big\| \chi_a\big( f(H_G) - f(H^{\tau}_G) \big) \chi_b \big\|_p  \big] 
		\leq C|\tau|^{s/2}			\NTV(f) \e^{-\mu(|a|+|b|)}.
	\end{equation} 
\end{theorem}

\begin{remark}
On the right hand side of \eqref{eq:TrClBounds1Stat} a linear $\tau$-dependence seems to be optimal. However, our proof does not allow for this.
\end{remark}

\noindent
The method we use in the proof of Theorem \ref{th:TrClBounds1} can also be applied to treat 
a perturbation by a boundary condition.  

\begin{theorem}
	\label{th:TrClBounds3}
	Fix $p>0$ and let $I \subset \FMB$ be a compact interval. Then there exist finite constants $C,\mu>0$ 
	such that the following holds: for all open subsets $G\subset \wtilde{G}\subseteq \R^d$ with 
	$\dist(\partial\wtilde{G},\partial G)\geq 1$, for all $a,b\in \R^d$ such that 
	$Q_a\cap {G}\ne \emptyset$ or $Q_b\cap {G}\ne \emptyset$, for all $\tau\in[0,1]$ and 
	all $f \in \mathcal{F}_I$ we have 
	\begin{equation}
		\label{eq:TrClBounds3Stat}
		\mathbb{E} \big[\big\| \chi_a\big(f(H_{G}^{\tau}) - f(H^{\tau}_{\wtilde{G}})\big)\chi_b \big\|_p  \big] 
		\leq C  \, \NTV(f)
		\e^{-\mu [\dist(a, \partial {G}) + \dist(b, \partial {G})]}.
	\end{equation} 
\end{theorem}

\begin{remark}
	The separate exponential decay of \eqref{eq:TrClBounds1Stat} in $a$ and $b$ reflects that 
	the operator $f(H_G)-f(H_G^{\tau})$ is exponentially small away from the support of $H_G-H_G^{\tau}$. 
	An analogous comment applies to \eqref{eq:TrClBounds3Stat}.
\end{remark}

\noindent
A consequence of Theorem~\ref{th:TrClBounds1} is

\begin{corollary}
	\label{cor:TrClBounds2}
	Fix $p,q>0$ and $\tau \in [0,1]$. Let $I \subset \FMB$ be a compact interval. 
	Then there exists a finite constant $C_{\tau}>0$ such that for all $G\subseteq\Rn$ open and $f \in \mathcal{F}_I$ 
	\begin{align}
	\label{eq:TrClBounds2Stat}
	\big(\mathbb{E} \left[ \| f(H_G) - f(H^{\tau}_G) \|_p^q  \right]\big)^{1/q} &\leq C_{\tau}\,  \NTV(f).
	\end{align}
	The constant $C_{\tau}$ vanishes algebraically in $\tau$ as $\tau \downarrow 0$. In particular, for $f \in \mathcal{F}_I$ 
	\beq\label{supertrace}
	f(H_G) - f(H^{\tau}_G) \in \S^p
	\eeq
	for any $p>0$ almost surely. 
\end{corollary}

\begin{remark}
	To obtain trace class properties as in \eqref{supertrace} for the pair $(-\Delta,-\Delta+V)$, where $V$ is a short-range scattering potential, $f$ needs to be in an appropriate Besov space  \cite{MR3314510}. This implies that $f$ has to be more regular than being continuous. In the region of complete localisation for random Schr\"odinger operators, even discontinuous functions $f$ are allowed in \eqref{supertrace}, as we traded regularity of the function $f$ for regularity of the resolvent \eqref{eq:DefFMB}. 
\end{remark}

\noindent
Theorem \ref{th:TrClBounds1} and Theorem \ref{th:TrClBounds3} together imply

\begin{theorem}
	\label{Thm:prod}
	Fix $p,q>0$ and let $I \subset \FMB$ be a compact interval. 
	Then there exist finite constants $C,\mu>0$ such that for all $f\in \mathcal F_I$, $\tau\in[0,1]$  
	and $L>0$ 
	\begin{equation}
		\label{thm:prod:eq1}	
		\Big(\mathbb{E} \left[\,\big\|\big( f(H_L) - f(H_L^\tau)\big) - 
		\big(f(H) - f(H^\tau)\big)\big\|_{p}^q\right]\Big)^{1/q}
		\leq C\NTV(f) \e^{-\mu L}. 
	\end{equation}
	Moreover, given a sequence $(L_n)_{n\in\N} \subset (0,\infty)$ of lengths with $L_{n}/\ln n \to \infty$ 
	as $n \to\infty$, then for every $\gamma\in (0,\mu)$ the convergence 
	\begin{equation}
		\label{eq:ConvergencePi}
	  \lim_{n\to\infty} \e^{\gamma L_n} \big\|\big(f(H_{L_n}) - f(H_{L_n}^{\tau})\big) - 
	  \big(f(H) - f(H^\tau)\big)\big\|_{p} = 0
	\end{equation}
	holds almost surely. 
\end{theorem}

\begin{remark}
	\label{fast-convergence}
	If $(X_{n})_{n\in\N} \subset (0,\infty)$ is a sequence of random variables with summable expectation, 
	$(\E [X_{n}])_{n\in\N} \in \ell^{1}(\N)$, then $\lim_{n\to\infty} X_{n} =0$ holds almost surely. 
	This elementary fact shows that \eqref{eq:ConvergencePi} follows directly from \eqref{thm:prod:eq1}.
\end{remark}

\noindent
Corollary~\ref{cor:TrClBounds2} or Theorem~\ref{Thm:prod} offer a way to detect spectral phases 
of random Schr\"odinger operators other than the region of complete localisation.

\begin{corollary}
	\label{phase-transition}
	Let $E \in \Sigma$ and $p>0$, then
	 
	\begin{equation}
		\label{nonFMBcrit}
 		\sup_{L>0} \E \big[\|\id_{(-\infty,E]}({H}_{L}) - \id_{(-\infty,E]}({H}^\tau_{L})\|_{p} \big] = \infty \quad\; \Longrightarrow \quad\;
		E \notin \FMB.
	\end{equation}

\end{corollary}

\begin{remark}
	Corollary~\ref{phase-transition}  provides a reasonable condition to conclude that $E\notin \FMB$ because it is met in the case of absolutely 			continuous spectrum and non-trivial scattering. 
	This can be seen as follows:
	assume $W\ge 0$ and observe that
	the finite-volume spectral shift operator $T(E,{H}_{L}, {H}^\tau_{L})$ is finite rank for all 
	$E\in\R$, hence even super trace class. 
	An iteration of \eqref{eq:square-products} then implies the representation
	\begin{align}
		\label{eq:shift-product}
		 \|\id_{(-\infty,E]}({H}_{L}) - \id_{(-\infty,E]}({H}^\tau_{L})\|_{2n}^{2n} & =  \|\id_{(-\infty,E]}({H}_{L})  
			\id_{(E,\infty)}(H^\tau_{L}) \id_{(-\infty,E]}({H}_{L})\|_{n}^{n} \notag\\
		& \quad\; + \|\id_{(E,\infty)}({H}_{L}) \id_{(-\infty,E]}({H}^\tau_{L}) \id_{(E,\infty)}({H}_{L})\|_{n}^{n}
	\end{align} 
	for every $n\in \N$.
	Now, suppose we knew there exists a 
	spectral  interval $J \subset\R$ such that $H$ has absolutely continuous spectrum in $J$ almost surely. 
	It follows from \eqref{eq:shift-product} and \cite[Thm.\ 3.4]{artAOC2015GKMO} that, almost surely given any $p>0$, 
	$\|\id_{(-\infty,E]}({H}_{L}) - \id_{(-\infty,E]}({H}^\tau_{L})\|_{p}$ diverges at least logarithmically in $L$ for a.e.\ $E\in J$ at which 
	there is non-trivial scattering, i.e.\  where the quantity in \eqref{gamma} below is strictly positive. 
\end{remark}

%
\subsection{Application 1: The spectral shift function (SSF)}\label{subsec:ssf}

In this subsection we apply the previous findings to prove convergence of the finite-volume
spectral shift functions for energies in $\FMB$. 
Moreover, we show that the infinite-volume spectral shift function, the trace of 
the spectral shift operator and the index of the corresponding pair of Fermi projections coincide. 
In addition, we establish H\"older continuity and positivity of the disorder-averaged spectral shift function 
in $\FMB$.

Let $A$ and $B$ be self-adjoint operators in a Hilbert space. 
For any given $E \in \R$, we call 
\begin{equation}
	\label{def:T(E)}
 	 T(E,A,B):= \id_{(-\infty,E]}(A) - \id_{(-\infty,E]}(B)
\end{equation}
\emph{spectral shift operator}. Being the difference of two projections, we recall that 
\begin{equation}
	\label{eq:T<1}
 	\|T(E,A,B) \| \le 1.
\end{equation}
Now, we assume further that $A$ and $B$ are bounded from below 
and that $\e^{-A}- \e^{-B}\in \S^1$. 
Then there exists a unique function $\xi := \xi(\,\pmb\cdot\,,A,B)\in L_{\text{loc}}^{1}(\mathbb{R})$, the 
\emph{spectral shift function (SSF)}, such that the equality 
\begin{equation}
	\label{eq:Definition1SSF}
	\Tr \big(f(A)-f(B)\big) = - \int_{\mathbb{R}}  \d\lambda\, f'(\lambda)\,\xi(\lambda,A,B) 
\end{equation}
holds for all test functions $f\in C^{\infty}(\mathbb{R})$ with $\lim_{\lambda\to\infty}f(\lambda)=0$ 
and $\supp (f')$ compact. We refer to, e.g., \cite{yafaev1992mathscattering} for details and further properties. 
Finally, given two self-adjoint projections $P$ and $Q$ such that 
$\pm1 \notin\sigma_{\text{ess}}(P-Q)$, the essential spectrum of $P-Q$, 
we define their \emph{index} according to \cite[Prop.\ 3.1]{artIND1994AvSeSi} as 
$\ind(P,Q) := \dim\ker (P-Q-1) - \dim\ker(P-Q+1)$. Further background on the index can be found in 
\cite{artIND1994AvSeSi}. 
In the particular case of the spectral projections $P=\id_{(-\infty,E]}(A)$ and $Q=\id_{(-\infty,E]}(B)$, we write
\begin{align}
	\label{def:index}
	\theta(E,A,B) &:= 	\ind \big\{ \id_{(-\infty,E]}(A), \id_{(-\infty,E]}(B) \big\} \notag\\
	& \phantom{:}= \dim \ker\big( T(E,A,B)-\id\big) - \dim \ker\big( T(E,A,B)+\id\big).
\end{align}
For $L>0$ and $A=H_{(L)}$ and $B=H_{(L)}^{\tau}$, we recall from 
Corollary \ref{cor:TrClBounds2} that 
\beq
	\label{eq:Ttraceclass}
	T(E,H_{(L)}, H^\tau_{(L)}) \in \S^{1} \qquad \text{almost surely for every $E\in\FMB$}.
\eeq
Hence, the index $\theta(E,H_{(L)},H^\tau_{(L)})$ is well defined almost surely for every $E\in\FMB$
and every $L>0$. On the other hand, we infer from \cite[Thm.~1]{MR2200269} that 
$\e^{-H_{(L)}} - \e^{-H^\tau_{(L)}} \in\S^{1}$ holds almost surely for all $L>0$. 
The SSF $\xi(\,\pmb\cdot\,, H_{(L)}, H^\tau_{(L)})$ is therefore well-defined 
almost surely as a function in $L_{\text{loc}}^{1}(\mathbb{R})$. 

For energies outside the essential spectrum of a pair of operators, 
it is known that all three quantities, the SSF, 
the index and the trace of the spectral shift operator, are well-defined and 
coincide, see \cite[Prop. 2.1]{art2009PushnIndex}. 
This extends to the region of complete localisation:

\begin{theorem}
	\label{th:SSFandTransOp}
	Let $\tau \in [0,1]$. Then,
	\begin{nummer}
	\item
		\label{th:SSFandTransOpi}
		the SSF, the trace of the shift operator and the index 
	  coincide, i.e.
		\begin{equation}
			\label{th:SSFandTransOpstat1}
			\xi(E,H,H^\tau)  = 	\Tr\big( T(E,H,H^\tau)\big)
			=\theta(E,H,H^\tau) \qquad \text{ for a.e. } E\in \FMB
		\end{equation}
		holds almost surely.
	\item
		\label{th:SSFandTransOpii}
		given a compact interval $I \subset \FMB$, there exist constants $C,\mu >0$ such that the bound
		\begin{equation}
			\label{th:SSFandTransOpstat3}
			\mathbb{E}\left[\l| \xi(E,H_{L},H^\tau_{L})-\xi(E,H,H^\tau)\r| \right]
			\leq C \e^{-\mu L} \qquad \text{ for a.e. } E\in I
		\end{equation}		
		holds for all $L>0$. Moreover, given a sequence of lengths with $L_{n}/\ln n \to \infty$ 
		as $n\to\infty$, the convergence
		\begin{equation}
			\label{th:SSFandTransOpstat2}
			\lim_{n\to\infty} \xi(E,H_{L_n},H^\tau_{L_n})=  \xi(E,H,H^\tau)\qquad \text{ for a.e. } E\in I
		\end{equation}
		holds almost surely.
	\end{nummer}
\end{theorem}

\begin{remarks}
\item
	\label{rem:right-equality-clear}
	It follows from \cite[Thm.\ 4.1]{artIND1994AvSeSi} that \eqref{eq:Ttraceclass} implies
	\beq
		\label{eq:trShift-index}
		\Tr \big( T(E,H_{(L)}, H^\tau_{(L)}) \big) = \theta(E,H_{(L)},H^\tau_{(L)})
	\eeq
	almost surely for every $E\in\FMB$ and $L>0$. The news in Theorem~\ref{th:SSFandTransOpi} 
	is the left equality. The right equality shows that all quantities are integer-valued. 
\item
 	In general, one cannot expect the SSF and the index to coincide 
	as soon as they are both well defined. For example, if $E$ lies within the absolutely continuous 
	spectrum of two operators, both the index and the SSF may 
	be well defined. But in this case, the SSF cannot be expected to be integer-valued, 
	whereas the index is by definition. For the very same reason, the convergence \eqref{th:SSFandTransOpstat2} cannot hold in general. 
	 The finite-volume spectral shift function is 
	given by 
	\begin{equation}
		\label{eq:xi-T}
 		\xi(E,H_{L},H_{L}^{\tau}) = \Tr \big(T(E, H_{L},H_{L}^{\tau}) \big) 
		= \Tr \big(\id_{(-\infty,E]}(H_{L})\big) - 	\Tr \big(\id_{(-\infty,E]}(H_{L}^{\tau})\big)	
	\end{equation}
	for all $E\in\R$ and, thus, integer-valued by definition. The infinite-volume quantity need not be. 
  However, for general Schr\"odinger operators $-\Delta+V_0$ and perturbations by  
  bounded, compactly supported potentials, vague convergence of the finite-volume SSF 
  is known \cite{MR2596053}. 
\item 
	Theorem \ref{th:SSFandTransOpii} is formulated in terms of the SSF. 
	It holds for the index as well, even without having to exclude a Lebesgue null set of exceptional energies. 
	Under the assumptions of \eqref{th:SSFandTransOpstat2} we obtain
	\begin{equation}
		\label{eq:SSFandIndex}
		\lim_{n\to\infty}\theta(E,H_{L_n},H^\tau_{L_n}) = \theta(E,H,H^\tau) \quad\text{almost surely for \emph{every} $E\in \FMB$}.
	\end{equation}
	This follows directly from \eqref{eq:trShift-index} and applying Theorem~\ref{Thm:prod} to the function $f=\id_{(-\infty, E]}$.
\end{remarks}

\noindent
An immediate consequence of Theorem \ref{th:SSFandTransOpii} is

\begin{corollary}
	\label{cor:bbdfiniteSSF}
	Let $\tau\in [0,1]$. Given a sequence of lengths with $L_{n}/\ln n \to \infty$ as $n\to\infty$, then 
	\beq
		\label{bdd:finiteSSF}
		\sup_{n\in\N} |\xi(E,H_{L_n},H_{L_n}^{\tau})| <\infty \qquad \text{ for a.e. } E\in \FMB
	\eeq
	holds almost surely.
\end{corollary}

\begin{remark}
	Bounds like \eqref{bdd:finiteSSF} are unknown for general deterministic continuum Schr\"odinger operators 
	in $d\ge 2$. 
	In the case of the pair $-\Delta$ and 
	$-\Delta+W$ with $0\leq W\in L^\infty_c(\R^d)$, a partial converse is shown in 
	\cite{artSSF1987Kir}: for any diverging sequence $(L_n)_{n\in\N}$ there exists a dense 
	subset $I\subseteq (0,\infty)$ such that for all $E\in I$ 
	\beq
		\sup_{n\in\N} \xi(E,-\Delta_{L_{n}},-\Delta_{L_{n}}+W) =\infty.
	\eeq
	In contrast, $L^{p}$-bounds for finite- or infinite-volume SSF of deterministic 
	continuum Schr\"odinger operators are known, see for example \cite{MR2200269, MR1824200, MR1945282}. Moreover, general upper 			bounds on the SSF are proved in \cite{MR1208792} assuming that a trace-class limiting absorption principle holds.
	In the case of random Schr\"odinger operators and for a particular choice of the perturbation, 
	\cite{MR2352262} show that the disorder-averaged finite- and infinite-volume SSF 
	are locally bounded uniformly in the system size. 
\end{remark}

\noindent
The next lemma implies H\"older continuity of the averaged trace of the shift operator in the region of complete localisation. 

\begin{lemma}
	\label{lem:HoelderCont}
	Let $I\subset\FMB$ be a compact interval and $\alpha \in (0,1)$.
	Then there exists a finite constant $C>0$ such 
	that for all $L>0$, all $\tau \in [0,1]$ and all $E,E'\in I$
	\begin{equation}
 		\E \big[\big\| T(E,H_{(L)},H_{(L)}^{\tau})- T(E',H_{(L)},H_{(L)}^{\tau}) \big\|_{1} \big] 
		\leq C |E-E'|^{\alpha}.
	\end{equation}
\end{lemma}

\begin{remark}
  Choosing representatives of the SSF as in Theorem \ref{th:SSFandTransOpi} and \eqref{eq:xi-T},
   we infer from Lemma~\ref{lem:HoelderCont} that for all $E,E'\in I$	
  \begin{equation}
 		\E \big[\big\vert \xi(E,H_{(L)},H_{(L)}^{\tau})-\xi(E',H_{(L)},H_{(L)}^{\tau}) \big\vert \big] 
		\leq C |E-E'|^{\alpha}.
	\end{equation}
\end{remark}

\noindent
Lemma~\ref{lem:HoelderCont} will be used in the proof of the next theorem which establishes 
pointwise positivity of the averaged SSF 
in the region of complete localisation for sufficiently positive perturbations. This will also be useful
in the following section.

\begin{theorem}
	\label{lem:SSFcont}
	Assume that the Lebesgue density $\rho$ of the single-site distribution, which is specified in \ref{assV1},
	satisfies 
$		\rho_-:= \essinf_{x\in[0,1]} \rho(x) > 0. $

	Let $\tau \in (0,1]$ and assume that the perturbation $W$ and the bump function $u_0$ satisfy 
	\beq
		\label{bounds:u0}
		W \geq C u_0\geq c \id_{B_\delta(x)}
	\eeq
	with constants $c,C, \delta>0$ and some $x\in\Lambda_1$. Then 
	\begin{equation}
		\E\l[\xi(E,H,H^{\tau})\r]>0 \qquad \text{ for a.e. } E\in \FMB\cap\, \{E'\in\R :\, \cN'(E')>0\},
	\end{equation}
	where $\cN'$ is the density of states of $H$.
\end{theorem}

\begin{remark}
	\label{rem:SSFcont}
	It is proven in \cite{DGHKM2016} that under the assumptions of Theorem \ref{lem:SSFcont} we have 
	\begin{equation}
		\label{eq:dospos}
		\essinf_{E\in I} \cN'(E) > 0
	\end{equation}
	for every compact interval $I\subset \FMB\cap \Int\big(\sigma(H_0) + [0,C_{u,-}]\big)$, 
	where $C_{u,-}$ is the constant from \eqref{uk-bounds} and $\text{int} (A)$ stands for the interior of a set
	$A\in\Borel(\R)$. This implies that 
	\beq
		\E\l[\xi(E,H,H^{\tau})\r]>0 \qquad \text{ for a.e. } E\in \FMB\cap \Int\big(\sigma(H_0) + [0,C_{u,-}]\big). 
	\eeq
\end{remark}

%
\subsection{Application 2: Anderson orthogonality}\label{subsec:aoc}

We apply the Schatten-class estimates from Section \ref{sect:TraceClassEst} to analyse Anderson
orthogonality in the region of complete localisation. Let $\tau>0$ be fixed for now. 
The finite-volume operators $H_L$ and $H_L^\tau$ have discrete spectrum, and we denote their eigenvalues by
$\lambda_1^L\le \lambda_2^L\le\cdots$  and
$\mu_1^L\le \mu_2^L \le \cdots$, 
respectively, ordered by magnitude and repeated according to their multiplicities. The corresponding 
orthonormal bases of eigenfunctions are $(\varphi_k^L)_{k\in\N}$ and $(\psi_k^L)_{k\in\N}$, 
respectively. For $N\in\N$ we consider two non-interacting $N$-Fermion systems
in the box $\Lambda_L$ with single-particle Schr\"odinger operators $H_L$ and $H^\tau_L$. 
The corresponding $N$-particle operators are
\begin{equation}
  \boldsymbol{H}^{(\tau)}_{N,L} :=
  \sum_{j=1}^N \underbrace{\id\otimes\dotsm\otimes \id}_{(j-1) \text{times}} \otimes H_L^{(\tau)} 
  \otimes \underbrace{\id\otimes\dotsm\otimes \id}_{(N-j) \text{times}}.
\end{equation}
They act on the totally antisymmetrised tensor product $\H_{N,L}:=\bigwedge_{j=1}^N L^2(\Lambda_L)$ 
with standard scalar product $\langle\,\pmb\cdot\,, \,\pmb\cdot\,\rangle_{\H_{N,L}}$.
The respective ground states of $\boldsymbol{H}_{N,L}$ and $\boldsymbol{H}^{\tau}_{N,L}$ are then 
given by the totally antisymmetrised and normalised tensor products
$\boldsymbol{\Phi}_{N,L} := \varphi_1^L\wedge\dotsm\wedge\varphi_N^L$ and
$\boldsymbol{\Psi}_{N,L} := \psi_1^L\wedge\dotsm\wedge\psi_N^L$.
The \textit{modulus} of their scalar product is given by
\begin{equation}
	\label{intro:defS}
	S_{N,L} := \big|\big\<\boldsymbol{\Phi}_{N,L},\boldsymbol{\Psi}_{N,L}\big\>_{\H_{N,L}}\big|
	= \left|\det \begin{pmatrix}
          \<\varphi^L_1,\psi^L_1\>& \cdots & \<\varphi^L_1,\psi^L_N\>\\
          \vdots                  &        &  \vdots \\
          \<\varphi^L_N,\psi^L_1\>& \cdots &  \<\varphi^L_N,\psi^L_N\>
    \end{pmatrix} \right|,
\end{equation}
where the equality follows from the Leibniz formula for determinants, and the scalar product appearing 
in the matrix entries is the one on the single-particle Hilbert space $L^{2}(\Lambda_{L})$. We note that 
$S_{N,L}$ depends on the coupling constant $\tau$, but we suppress it in our notation.
Our subsequent analysis relies on 
rewriting $S_{N,L}$ in terms of a Fredholm determinant. We recall that, given an operator $K\in \S^1$ 
with eigenvalues $(b_n)_{n\in\N}$, listed according to their algebraic multiplicities, 
the Fredholm determinant 
$\det( \id- K) := \prod_{n\in\N} ( 1- b_n)$
is well defined by the trace-class assumption on $K$ \cite[Sect.\ XIII.17]{reed1980methods4}. 

For fixed $E\in\R$, which will be referred to as the \emph{Fermi energy} in this context, 
we define the $L$-dependent particle number 
\begin{equation}
	\label{eq:DefParticleNumberSeq}
	N_{L}(E) := \Tr \big(\id_{(-\infty,E]} (H_L)\big).
\end{equation}
The limit $L \to\infty$ is a realisation of the thermodynamic limit with particle density given by 
the integrated density of states $\lim_{L\to\infty} N_{L}(E)/|\Lambda_L| = \cN(E)$ of the Hamiltonian $H$.
We define the \emph{finite-volume ground-state overlap}
\begin{equation}
\label{eq:DefOverlap}
	 S_L(E):= \left\{\begin{array}{c@{\quad\text{~if}\quad}l} S_{N_{L}(E),L}, & N_{L}(E) \in\N, \\[.5ex] 
		 1, & N_{L}(E)=0.\end{array}\right.
\end{equation}
Following physics terminology, we say that \emph{Anderson orthogonality} occurs at an energy $E$, if $S_{L}(E)$ vanishes as $L\to\infty$. 
Despite of the relevance in physics, mathematical results of this phenomenon are only rather recent and refer only to scattering systems: if the Fermi energy belongs to the absolutely continuous spectrum of the operator $H$, then there is algebraic decay of $S_{L}(E)$ as $L\to\infty$ \cite{artAOC2014KOS, artAOC2014GKM, artAOC2015GKMO, MR3376020, MR3405952}.

We define the \emph{infinite-volume ground-state overlap} as
\begin{equation}
	\label{fredholm2}
	S(E):= \left\{\begin{array}{c@{\quad}l} \det \big( \id - T(E,H,H^\tau)^2 \big)^{{1/4}}, 
		&\text{~ if } \;T(E,H,H^\tau)\in \mathcal{S}^2, \\[.5ex]
		0,  &\text{~ otherwise.} \end{array} \right.
\end{equation}
We note that if $E\in \FMB$, then $T(E,H,H^\tau)\in\mathcal{S}^p$ almost surely for any $p>0$ by 
Corollary \ref{cor:TrClBounds2}. 

\begin{theorem}
	\begin{nummer}
	\item
		\label{th:IntMainResi}
		For all $E\in\FMB\cap \Int(\Sigma)$ the convergence
		\begin{equation}
			\label{eq:ResultNoAocStat1}
			\lim_{L\to\infty} \mathbb{E}\big[S_{L}(E)\big] = \mathbb{E}\big[S(E)\big]
		\end{equation}
		holds. If, in addition, $L_{n}/\ln n \to \infty$ as $n\to\infty$, then the pointwise convergence
		\begin{equation}
			\label{eq:ResultNoAocStatPointw}
			\lim_{n\to\infty} S_{L_n}(E) = S(E)
		\end{equation}
		holds almost surely. 
	\item
		\label{th:IntMainResiprime}
		Let $E\in\FMB$. Then 
		\begin{equation}
			\label{eq:ResultNoAocStat2}
			S(E)=0 \iff 1\in\sigma\big(T(E,H,H^\tau)^2\big)
		\end{equation}
		holds almost surely.
	\item
		\label{th:IntMainResii}
		Let  $I\subset \FMB$ be a compact interval.  
		Then we have
		\beq\label{AbsAoc2}
			\lim_{\tau\downarrow 0} \inf_{E\in I}\,\E\big[S(E) \big] =  1.
		\eeq
	\item
		\label{th:IntMainResiii}
		Suppose that the assumptions of Theorem \ref{lem:SSFcont} are fulfilled. Then
		\begin{equation}\label{AbsAoc}
			\PP\big[S(E)= 0\big] > 0 \qquad\text{ for a.e. } E\in \FMB \cap \,\{ E'\in\R: \cN'(E') >0\},
		\end{equation} 
		where $\cN'$ is the density of states of $H$.
	\end{nummer}
	\label{th:IntMainRes}
\end{theorem}

\begin{remarks}
\item 
The set of energies for which \eqref{AbsAoc} holds is not empty for our model, see Remark \ref{rem:SSFcont}.
\item 
	Theorem \ref{th:IntMainResii} shows that there is no relic of Anderson orthogonality when switching off the 
	perturbation. A similar phenomenon is expected to occur in the large-disorder limit. 
\item 
	Theorem \ref{th:IntMainResiii} shows that Anderson orthogonality is not absent almost surely in the 
	region of complete localisation, contrary to established wisdom 
	in the physics literature \cite{PhysRevB.65.081106}. 
	Only very recently, physical reasoning and numerical evidence was put forward that Anderson orthogonality 
	does occur in the region of complete localisation \cite{Khemani-NaturePhys15,PhysRevB.92.220201}. 
\item
	For sign-definite perturbations $W$, the condition \eqref{eq:ResultNoAocStat2} for Anderson orthogonality to 
	occur in the region of complete localisation is 
	equivalent to a non-vanishing index,
	\begin{equation}
	\label{AOCequiIndex}
		S(E)=0 \iff \theta(E,H,H^\tau) \neq 0
	\end{equation}
	almost surely for every $E\in\FMB$.
	This follows from the fact that for a sign-definite perturbation, $1$ and $-1$ cannot both be eigenvalues of $T(E,H,H^\tau)$ and therefore $1\in\sigma\big(T(E,H,H^\tau)^2\big) \iff \theta(E,H,H^\tau) \neq 0$.
\item
	Let $I \subset \FMB$ be compact. Then, Theorem \ref{th:IntMainRes} \itemref{th:IntMainResii} and
	\itemref{th:IntMainResiii} imply that there exists $\tau_I>0$ such that for all $0<\tau<\tau_I$ 
	and a.e.\ $E\in I \cap \,\{ E'\in\R: \cN'(E')\} >0$
	\beq
		0<\PP\big[S(E)=0\big]<1.
	\eeq
	Thus, $S(E)$ is not almost surely constant, and both presence and absence of Anderson orthogonality occur with positive probability for energies in $\FMB$.
	\item
	\label{rm:expfastconv}
	Let $E\in \FMB \cap \Int(\Sigma)$ such that 
	$\theta(E,H,H^\tau) = 0$. 
	In this case our proof of 
	Theorem~\ref{th:IntMainResi} provides exponential speed of convergence of the overlap in 
	\eqref{eq:ResultNoAocStatPointw}: we conclude from
	\eqref{eq:ProofIndexZero}  and \eqref{eq:ConvergencePi} that, given a sequence of 
	lengths with $L_{n}/\ln n \to \infty$ as $n\to\infty$, there exists a constant $c>0$ such that
	\beq
		\label{expfastconv}
		\lim_{n\rightarrow\infty} \e^{cL_n}\big|S_{L_n}(E) - S(E)\big| = 0
	\eeq
	holds almost surely. According to the physical arguments in \cite{PhysRevB.92.220201}, 
	exponential speed of convergence occurs with positive probability for $E\in \FMB$ even if 
	$\theta(E,H,H^\tau) \neq 0$. 
	More precisely, it is argued there that $\E[\ln S_{L}(E)] \sim -L$ for large $L$. 
  This phenomenon is named
	\emph{statistical Anderson orthogonality} in \cite{PhysRevB.92.220201} to distinguish it from the 
	usual algebraic decay for energies in the scattering regime, see also the next remark. 
\end{remarks}

\noindent
Similarly to Corollary \ref{phase-transition}, Theorem~\ref{th:IntMainResi} and~\itemref{th:IntMainResii} offer a way to detect spectral phases 
of random Schr\"odinger operators other than the region of complete localisation.

\begin{corollary}
	\label{rm:notFMB}
	Let $E \in \Int(\Sigma)$. Assume there exists a 
	null sequence of coupling constants $(\tau_{k})_{k\in\N} \subset (0,1)$ and a sequence of 
	lengths $(L_{n})_{n\in\N}$ with $L_{n}/\ln n \to \infty$ as $n\to\infty$ such that the $\tau_{k}$-dependent 	overlap obeys $\lim_{n \to \infty} S_{L_{n}}(E) =0$ almost surely for every $k\in\N$, 
	then $E \notin \FMB$. 
\end{corollary}

\begin{remark}
	The assumption of Corollary~\ref{rm:notFMB} is reasonable. This can be seen as follows:
	assume $W\ge 0$ and 
	suppose we knew there exists a spectral interval $J \subset \Sigma$ 
	such that $H$ has absolutely continuous spectrum in $J$ almost surely. Let $(L_{n})_{n\in\N}$ be a 
	sequence of lengths such that $L_{n} / \e^{n^{\alpha}} \to \infty$ as $n\to\infty$ for some $\alpha>1$. 
	Then, Theorem 2.2 in \cite{artAOC2015GKMO}, see also \cite{artAOC2014GKM}, applies realisationwise 
	to the operators $H$ and $H^{\tau}$ for any $\tau\in (0,1]$, and we infer that
	almost surely and for a.e.\ energy $E\in J$ 
	\begin{equation}
		\label{AOC}
		S_{L_{n}}(E) \le  L_{n}^{-\gamma(E)/2 + o(L_{n}^{0})} \qquad\text{as $n\to\infty$}.
	\end{equation}
	The decay exponent 
  \begin{equation}
  	\label{gamma}
		\gamma(E) :=\pi^{-2} \big\| \arcsin|(\id-\mathbb{S}_E)/2| \big\|_2^2 
	\end{equation}
	relates to the energy-dependent scattering matrix $\mathbb S_{E}$ and is strictly positive if the 
	perturbation $W$ causes non-trivial scattering at energy $E$. 
	We remark that the super-exponential growth 
	of the length scales $(L_{n})_{n\in\N}$ avoids the necessity of passing to a subsequence 
	in \eqref{AOC} as is done in 
	\cite[Thm.\ 2.2]{artAOC2015GKMO}. This is so because it implies a sufficiently fast $L^{1}$-convergence in 
	the proof of \cite[Lemma 3.3(i)]{artAOC2015GKMO} so that subsequences can be avoided in that lemma. 
\end{remark}

%
%
\section{Proofs of the results from Section \ref{sect:TraceClassEst}}
\label{sec:4}

\subsection{Helffer--Sj\"ostrand formula}

In this section we prove a spatially localised variant of the Helffer--Sj\"ostrand formula which allows for jump discontinuities.

We consider functions $f\in \TV_c(\R)$. Given such $f$, we choose a cutoff function 
$\Xi\in C_{c}^{\infty}(\R^{2})$ with $\Xi\equiv 1$ in a neighbourhood $N_f$ of $\supp(f)\times \{0\} \subset \R^2$. 
We then define the complex Borel measure $\zeta_{f}$ on $\R^{2}$ by 
\begin{equation}
	\label{eq:LocHelfSjoDef1}
	\d\zeta_f(x,y) := \d f(x)\, \d y\, \Xi(x,y) + \d x\,\d y\, f(x) \left(\partial_x+\i\partial_y\right) \Xi(x,y),
\end{equation}	
where $\d f$ denotes Lebesgue-Stieltjes integration with respect to $f$. We write $|\zeta_{f}|$ for the 
total variation measure of $\zeta_{f}$.

\begin{lemma}
	\label{lem:LocHelfSjo}
	Let $f\in \TV_c(\R)$ and $\zeta_{f}$ as in \eqref{eq:LocHelfSjoDef1}. Let $K$ be a self-adjoint operator and 
	let $A$, $B$ be bounded operators on the Hilbert space $\mathcal{H}$. If
	\beq
		\label{eq:LocHelfSjoAssumpt}
		\int_{\R^2} \d|\zeta_{f}|(x,y)\,   \| A \, R_{x+\i y}(K) \, B \| <\infty,
	\eeq
	then 
	\begin{equation}
		\label{eq:LocHelfSjoStat}
		A f(K) B = \frac{1}{2\pi} \int_{\R^{2}} \d\zeta_{f}(x,y) \, A \, R_{x+\i y}(K) \, B \, 	
	\end{equation}
	holds, where the right-hand side is a Bochner integral with respect to the operator norm. 
\end{lemma}

\begin{remark}
	Lemma \ref{lem:LocHelfSjo} can be extended to appropriate Besov spaces, $B_{p,q}^s$, $1\leq p,q\leq \infty$ 
	and $0<s<1$, by using Dynkin's characterization of Besov spaces  in terms of quasi-analytic 
	extensions \cite{zbMATH03808305}. 
\end{remark}

Before we prove the lemma, we apply it to spatially localised functions of random Schr\"odinger operators which obey the a priori estimate \eqref{eq:AppAPriori}.

\begin{corollary}
	\label{cor:HSrandom}
 	Let $G\subseteq\R^d$ be open, $a,b\in G$ and let $H^{\tau}_G$ be defined as in \eqref{eq:TheOperator} and 	
	\eqref{def:perturbedoperator}. Then, for $f\in \TV_c(\R)$, the equality
	\begin{equation}
		\label{eq:LocHelfSjorandom}
		\chi_{a} f(H^{\tau}_{G}) \chi_{b} = \frac{1}{2\pi} \int_{\R^{2}} \d\zeta_{f}(x,y) \, \chi_{a} \, 
		R_{x+\i y}(H^{\tau}_{G}) \, \chi_{b}	
	\end{equation}
	holds almost surely.
\end{corollary}

\begin{proof}
	We verify that \eqref{eq:LocHelfSjoAssumpt} holds almost surely. There exists $\delta>0$ (independent of $f$)
	such that $\supp(\Xi) \subset \R \times [-\delta,\delta]$. For fixed $0<s<1$, we use the deterministic norm bound 
	$\| R_z(H^{\tau}_G)\|^{1-s}\leq 1/ |\text{Im} z|^{1-s}$, $z\in\C\setminus\R$, and estimate 
 	\begin{align}
		\label{eq:LocHelfSjo1}
		& \mathbb{E}\left[	\int_{\R^2} \d|\zeta_{f}|(x,y)\,  \| \chi_a\, R_{x+\i y}(H^{\tau}_G) \, \chi_b \| \right] 
		\leq  C_s  \int_{\R^2} \frac{\d|\zeta_{f}|(x,y)}{|y|^{1-s}},
	\end{align}
	where 
	$	C_s:= 
		\sup_{{(E,\eta)\in \supp(\zeta_f), \eta\neq 0}} \mathbb{E} 
		\left[ \Vert \chi_{a}R_{E+\i\eta}(H^{\tau}_G)\chi_{b} \Vert^{s} \right] <\infty $
   is finite by the a priori estimate \eqref{eq:AppAPriori}. The claim then follows from     
	\begin{align}
		\label{eq:zeta-int}
 		\int_{\R^2} \frac{\d|\zeta_{f}|(x,y)}{|y|^{1-s}} \le \frac{2}{s}\, \delta^{s} \big(\NTV(f) \|\Xi\|_{\infty} + 
		2 \|f\|_{1} \| |\nabla\Xi|\|_{\infty}\big). 
		\\[-1ex]	\tag*{\qedhere}
	\end{align}
\end{proof}

\begin{proof}[Proof of Lemma \ref{lem:LocHelfSjo}]
	We note that assumption \eqref{eq:LocHelfSjoAssumpt} implies that the right-hand side of 
	\eqref{eq:LocHelfSjoStat} is well defined as a Bochner integral with respect to the operator norm.  
	For $\varepsilon>0$, we split the integration into the set $\mathcal{C}_{\varepsilon} := \R \times [-\varepsilon,\varepsilon]$   
	and its complement in $\R^2$. Using dominated convergence and the integrability assumption 				             
	\eqref{eq:LocHelfSjoAssumpt}, we see that the integral over $\mathcal{C}_{\varepsilon} $
        	vanishes as $\varepsilon\to 0$. Therefore, we obtain
	that
	\begin{align}
		\label{eq:LocHelfSjo2}
		 \frac{1}{2\pi} \int_{\mathbb{R}^{2}}  \d\zeta_{f}(x,y) \,  A \, R_{x+\i y}(K) \, B 
		= \lim_{\varepsilon\downarrow 0}  A f_{\varepsilon}(K) B
	\end{align}
	in the operator-norm topology with 
	\begin{align}
		\label{eq:LocHelfSjo5}
 		f_{\varepsilon}(\lambda) 
		&:= \frac{1}{2\pi} \int_{\R^2  \setminus \mathcal{C}_{\varepsilon} }  \frac{\d\zeta_{f}(x,y)}{\lambda -x -\i y} 
			= \int_{\R} \d x \, \frac{1/\pi}{x^{2} + 1} \, f(\lambda+ \varepsilon x) 
	\end{align}
	for $\lambda\in\R$.  
	
	The second part of the proof deals with the new aspect that discontinuity points of $f$ challenge the convergence of $f_{\varepsilon}$ as $\varepsilon\downarrow 0$. Let 
	$\varphi,\psi \in\mathcal{H}$ and define the complex spectral measure $\mu_{\varphi,\psi} := \langle \varphi, A \id_{\raisebox{-.4ex}{$\boldsymbol\cdot$}}(K) B\psi\rangle$ of $K$. The functional calculus, \eqref{eq:LocHelfSjo5} and dominated convergence imply
\begin{equation}
	\label{eq:f-eps-weak}
	 \lim_{\varepsilon\downarrow 0} \langle \varphi, A f_{\varepsilon}(K) B\psi\rangle = 
	 \int_{\R}\d\mu_{\varphi,\psi}(\lambda) \, \lim_{\varepsilon\downarrow 0} \int_{\R} \d x\, 
	 \frac{1/\pi}{x^{2} + 1} \, f(\lambda+ \varepsilon x).
\end{equation}
We prove below that the set of discontinuity points of $f$ is a $\mu_{\varphi,\psi}$-null set. Using this, another application of dominated convergence in \eqref{eq:f-eps-weak} yields $\lim_{\varepsilon\downarrow 0} Af_{\varepsilon}(K)B = Af(K)B$ weakly, and the lemma follows. 

It remains to prove that $f$ is continuous $\mu_{\varphi,\psi}$-almost everywhere. 
Without loss of generality, we assume $\|\varphi\| = \|\psi\| =1$.
Since $f\in \TV_{c}(\R)$, it is right-continuous and has left limits at all points. Hence, the set 
$\mathcal{U} := \{\lambda\in\R: f \text{ not continuous in } \lambda \}$ consists of jump discontinuities only. Moreover, 
$\mathcal{U}$ is countable so that 
$ \mu_{\varphi,\psi}(\mathcal{U}) = \sum_{\lambda \in \mathcal{U}} \mu_{\varphi,\psi}(\{\lambda\})$.
We fix an arbitrary $\lambda\in\mathcal{U}$ and set $\delta f_{\lambda} := \lim_{\varepsilon\downarrow 0} [f(\lambda) - f(\lambda -\varepsilon)] \neq0$. We choose $y_{0} >0$ small enough such that $\Xi(\lambda,y)=1$ whenever $|y| \le y_{0}$. 
Assumption \eqref{eq:LocHelfSjoAssumpt} implies
\begin{align}
	\label{eq:ass-cons}
	\infty & >  \int_{\R^2} \d|\zeta_{f}|(x,y)\,   \| A \, R_{x+\i y}(K) \, B \| 
		\ge |\delta f_{\lambda}|  \int_{-y_{0}}^{y_{0}} \d y\, \| A \, R_{\lambda+\i y}(K) \, B \| \notag\\
	& \ge  |\delta f_{\lambda}|  \int_{-y_{0}}^{y_{0}} \frac{\d y}{|y|}\, |h(y)|,
\end{align}
 where $h(y) := \int_{\R}\d\mu_{\varphi,\psi}(\lambda') \, y/(\lambda' - \lambda -\i y)$. Dominated convergence implies 
that $\R \ni y \mapsto h(y)$ is continuous and that $\lim_{y\to 0} h(y) = \i\,\mu_{\varphi,\psi}(\{\lambda\})$.
Now, we assume that $\mu_{\varphi,\psi}(\{\lambda\}) \neq 0$. Then there exists $0<y_{1} \le y_{0}$ such that 
$|h(y)| \ge  |\mu_{\varphi,\psi}(\{\lambda\})|/2 $ whenever $|y| \le y_{1}$, and we conclude that \eqref{eq:ass-cons} 
yields a contradiction. Therefore we must have $\mu_{\varphi,\psi}(\{\lambda\})=0$, and $\mathcal{U}$ is a null set. This implies the desired continuity of $f$.
\end{proof}

%
\subsection{Proof of Theorem \ref{th:TrClBounds1} and Corollary \ref{cor:TrClBounds2}}
\label{pf:th:TrClBounds1}

We introduce a lower bound of the spectra of all considered random Schr\"odinger operators, 
\begin{align}
	\label{eq:DefE0}
	E_0 &:= \min \big\{\inf_{x\in\R^d}  \{V_0(x)\}, \inf_{x\in \R^d} \{V_0(x)  + W(x)\}\big\} \notag\\
	& \phantom:\leq \min\big\{ \min\{\sigma(H_G)\}, \min\{\sigma(H^\tau_G)\}\big\}.
\end{align}
Here, the inequality holds almost surely for all $\tau \in [0,1]$ and all $G\subseteq \R^d$ open. 

We begin with an operator-norm version of Theorem \ref{th:TrClBounds1}. Given a function $f\in \TV(\R)$, 
we use the notation $\supp(f')$ for the support of the complex measure defined by Lebesgue-Stieltjes 
integration with respect to $f$.

\begin{theorem}
	\label{th:TrClBounds1:OpNorm}
	Fix $0<s<1$ and a compact set $S \subset\FMB$. Then there exist finite constants $C,\mu>0$ such that  
	for all functions $f\in\BV_c(\R)$ with $\supp(f')\subseteq S$, for all $\tau\in[0,1]$, all open $G\subseteq\Rn$ 
	and all $a,b\in \R^d$ we have
	\begin{equation}
		\label{eq:OpBounds1Stat2}
		\mathbb{E} \left[ \left\| \chi_a\big( f(H_G) - f(H^\tau_G) \big) \chi_b \right\|  \right] 
		\leq C|\tau|^{s/2}\,\big( \|f\|_1+\NTV(f) \big) \e^{-\mu\left(|a|+|b|\right)}.
	\end{equation}	
\end{theorem}

\begin{proof} 
	We fix $f\in\BV_c(\R)$ such that $\supp(f')\subseteq S$. Let $S_{-} := \inf S$, $S_{+} := \sup S$ and 
	observe $\supp(f) \subseteq [S_{-},S_{+}]$. We choose a cutoff function $\Xi=\Xi_{f}$ subject to 
	\begin{MyDescription}
	\item[(P1)]{ 
		\label{K1}		
		$\Xi\in C_c^{\infty}(\R^2)$ with $0\leq \Xi \leq 1$ and $\|\partial_{x}\Xi\|_{\infty},\|\partial_{y} 
		\Xi\|_{\infty} \leq 3$,}			
	\item[(P2)]{ 
		\label{K2}
		$\supp(\Xi)\subseteq [S_- -1,S_+ +1] \times [-1,1]$,} 
	\item[(P3)]{ 
		\label{K3}
		$\Xi\equiv 1$ on $[S_- -1/2,S_+ + 1/2] \times [-1/2,1/2]$.}
	\end{MyDescription}
	Now fix $G\subseteq\Rn$ open and $a,b\in \R^d$. Let $\zeta_h$ be the complex Borel measure defined 
	in \eqref{eq:LocHelfSjoDef1}. We apply Corollary~\ref{cor:HSrandom} to the operators $f(H_G)$ and $f(H^\tau_G)$, choose some $0<s<1$ and use the norm bound 
	$\| R_z(H_{G}^{(\tau)})\|^{1-s/2}\leq 1/|\text{Im}z|^{1-s/2}$. 
	This gives the estimate
	\begin{align}
		\label{combes-thomas:eq10}
		\mathbb{E} &\left[ \big\| \chi_a\big( f(H_G) - f(H^\tau_G) \big) \chi_b \big\|  \right] \notag\\
		&	\leq \frac {2^{1-s/2}} {2\pi} \E \l[ \int_{\mathbb{R}^2} \frac{\d|\zeta_{f}|(x,y)}{|y|^{1-s/2}}  \, 
			\big\|\chi_a \big(R_{x+\i y}(H_G)-R_{x+\i y}(H^\tau_G)\big)\chi_b \big\|^{s/2} \r].
	\end{align}
	We introduce the subset of lattice points
	\begin{equation}
		\label{def:GammaW}
 		\Gamma_{W} := \big\{n\in\Z^{d}: \dist\big(n,\supp(W)\big)\leq 1/2)\big\}
	\end{equation}
	needed as centres to cover 
	$\supp(W)$ by closed unit cubes and deduce from the resolvent equation and the Cauchy-Schwarz inequality
	\begin{align}
		\label{combes-thomas:eq2}
		\E \bigg[ \int_{\mathbb{R}^2} \frac{\d|\zeta_f|(x,y)}{|y|^{1-s/2}}\, &\big\|\chi_a 
			\big(R_{x+\i y}(H_G)-R_{x+\i y}(H^\tau_G)\big)\chi_b \big\|^{s/2} \bigg] \notag \\
		& \leq \|\tau W\|_\infty^{s/2} \sum_{c\in {\Gamma_{W}}}  \int_{\mathbb{R}^2} \frac{\d|\zeta_f|(x,y)}{|y|^{1-s/2}} \; \mathbb{E} \big[ \|\chi_a R_{x+\i y}(H_G)\chi_c \|^{s}\big]^{1/2}
		  \notag \\
		& \hspace{3.5cm} \times \, \mathbb{E} \big[ \|\chi_c R_{x+\i y}(H^\tau_G) \chi_b\|^{s}\big]^{1/2}.
	\end{align}
	We infer from \eqref{eq:LocHelfSjoDef1}, \ref{K2} and \ref{K3} that 
	\begin{equation}
		\label{def:suppomega}
		\supp ( \zeta_f ) \subseteq \Big(\supp(f')\times [-1,1]\Big)\cup \Big([S_- ,S_+ ] \times \big([-1,-1/2] 
		\cup [1/2, 1]\big)\Big) 
		=: Z_1\cup Z_2.
	\end{equation}
	On the set $Z_2$ we estimate the right-hand side of \eqref{combes-thomas:eq2} by the Combes-Thomas estimate 
	stated in \cite[Cor.\ 1]{MR1937430} for deterministic Schr\"odinger operators. Even though this result is 
	formulated for Schr\"odinger operators on $L^2(\R^d)$, the argument extends to Schr\"odinger operators on 
	$L^2(G)$ for arbitrary $G\subseteq \R^d$ open -- see also  \cite{Shen2014}. Thus, there exist constants 
	$C_{0},\mu_{0}>0$, which are independent of $G\subseteq\R^d$, $\tau\in [0,1]$ and $\omega\in \Omega$, such that
	\begin{equation}
		\label{CombesThomas:cont}
		\|\chi_a R_z(H^{\tau}_G)\chi_b \| \leq C_{0} \e^{-\mu_{0} |z|^{-1}|a-b|}
	\end{equation}
	for all $z\in\C$ subject to $\dist(z,\sigma(H_G^{\tau})) \geq 1/2$. For $(x,y)\in Z_2$ we have 
	\beq\label{eq:distspec}
 		|x+ \i y| \leq  {|x| + |y|} \leq  \max\{|S_-|, |S_+|\} + 1 =: C_{S}
	\eeq
	and therefore we obtain finite constants $C_1>0$ and $\mu_{1} := 2s\mu_{0}/C_{S}>0$ such that
	\beq
		\label{Combeseq1}
		\sup_{(x,y)\in Z_2} \Big\{ \mathbb{E} \big[ \|\chi_a R_{x+ \i y}(H_G)\chi_c \|^{s}\big] \,
		\mathbb{E} \big[ \|\chi_c R_{x+\i y}(H^\tau_G) \chi_b\|^{s}\big] \Big\}
		\leq C_1 \e^{-\mu_{1}(|a-c| + |c-b| )}.
	\eeq

	On the set $Z_1$ we use the fractional-moment bounds \eqref{eq:DefFMB} for $H$ and $H^\tau$, which can be applied 
	because of $\supp( f') \subseteq \FMB$ and Remark~\ref{FMBalways2}. 
	Hence, there exist finite constants $C_{2}>0$ and $\mu_{2}>0$ such that
	\beq
		\label{Combeseq2}
		\sup_{(x,y)\in Z_1} \Big\{ \mathbb{E} \big[ \|\chi_a R_{x+\i y}(H_G)\chi_c \|^{s}\big] \,
		\mathbb{E} \big[ \|\chi_c R_{x+\i y}(H^\tau_G) \chi_b\|^{s}\big] \Big\}
		\leq C_{2} \e^{-\mu_{2}(|a-c| + |c-b| )}.
	\eeq 

	Collecting the estimates in \eqref{combes-thomas:eq10}, \eqref{combes-thomas:eq2}, \eqref{Combeseq1} and
	\eqref{Combeseq2}, we obtain finite constants $C_{3}>0$ and $\mu:=\min\{\mu_{1},\mu_{2}\}/2>0$, 
	which depend on $s$ and $S$ but are independent of $G$ and $\tau\in[0,1]$ such that 
	\begin{align}
		\label{combes-thomas:eq3}
		\mathbb{E} \left[ \big\| \chi_a\big( f(H_G) - f(H^\tau_G) \big) \chi_b \big\|  \right] 
		& \leq C_{3} \|\tau W \|_\infty^{s/2} \sum_{c\in {\Gamma_{W}}} \e^{-\mu(|a-c| + |c-b|)}  
			\int_{\R^2} \frac{\d|\zeta_f|(x,y)}{|y|^{1-s/2}}  \notag\\
		& \leq  C |\tau|^{s/2} \big(\|f\|_1+\NTV(f)\big) \e^{-\mu\left(|a|+|b|\right)}.
	\end{align} 
	To obtain the last inequality, we applied \eqref{eq:zeta-int}  and the estimate 
	\begin{align}
		\label{eq:expproduct}
		\sum_{c\in {\Gamma_{W}}} \e^{-\mu(|a-c| + |c-b|)} \le \e^{-\mu(|a| + |b|)} \sum_{c\in {\Gamma_{W}}} 
		\e^{2\mu|c|}. 
		\\[-2ex]	\tag*{\qedhere}
	\end{align}
\end{proof}

\begin{proof}[proof of Theorem \ref{th:TrClBounds1}]
	We set $I_{-}:= \min I$, $I_{+} := \max I$ and assume without restriction that $I_{+} \ge E_{0}$, 
	the lower bound \eqref{eq:DefE0} for the spectrum of $H_{G}^{\tau}$ (otherwise the statement is trivial). 
	We fix $G\subseteq\Rn$ open and $a,b\in \R^d$. 
	Since $\|A\|_{p} \le \|A\|_{1}$ for any Schatten norm with $p \ge 1$, we restrict ourselves to $0<p \le 1$. 
	
 	Let $f\in \mathcal F_I$. As we want to apply Theorem \ref{th:TrClBounds1:OpNorm}, 
	which requires functions of compact support, we introduce $h:= f \id_{[E_{0}-1, \infty)}$. 
	Then we have  $h\in \text{BV}_c(\R)$, $\NTV(h)\leq 2 \NTV(f)$, 
 	$f(H^{\tau}_G) =  h(H^{\tau}_G)$
	for every $\tau\in[0,1]$ and $\supp(h') \subseteq I\cup \{E_0-1\} \subset\FMB$. The last inclusion holds 
	by the definition of $E_{0}$ and the Combes-Thomas estimate \cite[Cor.\ 1]{MR1937430}, 
	which extends to finite-volume operators.
	
	We obtain for any $0<r<p$
	\begin{multline}
		\label{trclbound1:eq1}
		\big\| \chi_a\left( f(H_G) - f(H^\tau_G) \right) \chi_b \big\|_p 
		= \big\| \chi_a\left( h(H_G) - h(H^\tau_G) \right) \chi_b \big\|_p\\
		\leq  \big\|\chi_a\left( h(H_G) - h(H^\tau_G) \right) \chi_b \big\|^{r/p} 
		\big\|\chi_a\left( h(H_G) - h(H^\tau_G) \right) \chi_b \big\|_{p-r}^{1- r/p}.
	\end{multline}
	The adapted triangle inequality \eqref{lem:conkav} and the deterministic a priori estimate from Lemma \ref{lem:aPriori}
	imply
	\begin{align}
		\label{trclbound1:eq2}
		\big\|\chi_a\left( h(H_G) - h(H^\tau_G) \right) \chi_b \big\|_{p'}^{p'} \leq 
		\|\chi_a h(H_G) \chi_b\|_{p'}^{p'} + \|\chi_a h(H^\tau_G)  \chi_b \|_{p'}^{p'} 
		\leq C_{p'}  \|h\|^{p'}_\infty \notag\\
	\end{align}
	for every $0< p' \le 1$, where $C_{p'}$ depends also on $I_{+}$, but is uniform in the disorder, 
	$G$ and $\tau \in [0,1]$, and independent of $h$. 
	We apply \eqref{trclbound1:eq2} with $p'=p-r$ to estimate the expectation of \eqref{trclbound1:eq1} by 
	\begin{align}
		\label{TrCl-done}
		\mathbb{E} \Big[ \big\| \chi_a\left( f(H_G) - f(H^\tau_G) \right) \chi_b \big\|_p \Big]
 		\le C_{p-r}^{1/p} \, \| h \|_{\infty}^{1-s_1}  \,
		\mathbb{E} \Big[ \big\| \chi_a \big( h(H_G) - h(H^\tau_G) \big) \chi_b \big\| \Big]^{s_1}. \notag\\
	\end{align}
	Here, we introduced $s_1:=r/p < 1$ and also used Jensen's inequality.

	Now, we choose $0<s_2<1$ and apply Theorem \ref{th:TrClBounds1:OpNorm} with $S= I \cup \{E_{0} -1\}$ 
	to the expectation on the right-hand side of \eqref{TrCl-done}. This yields finite constants $C_{1}, \mu_{1} >0$ 
	(which depend on $s_{2}$ and $I$) such that 
	\begin{equation}
		\label{pf:Thm3.1equality3}
		\mathbb{E} \Big[ \big\| \chi_a \big( h(H_G) - h(H^\tau_G) \big) \chi_b \big\| \Big] \leq C_{1}| \tau|^{s_2/2}\,
		(I_+- E_0 + 2) \NTV(f)\, \e^{-{\mu_{1}}\left(|a|+|b|\right)},
	\end{equation}
	where we used 
	\begin{equation}
 	\label{eq:1normTV}
		\|h\|_1+\NTV(h) \leq  2(I_+ - E_0+ 2) \NTV(f).
	\end{equation}
	Inserting \eqref{pf:Thm3.1equality3} into \eqref{TrCl-done} and observing 
	\begin{equation}
		\label{eq:inftynormTV}
 		\|h\|_{\infty} \leq \NTV(h) \leq 2 \NTV(f),
	\end{equation}
	we obtain 
	\begin{align}
		\label{pf:Thm3.1equality4}
		\E \l[ \big\| \chi_a\left( f(H_G) - f(H^\tau_G) \right) \chi_b \big\|_p\r]
		\leq C |\tau|^{s_1s_2/2} \NTV(f) \e^{-s_1\mu_{1}\left(|a|+|b|\right)}
	\end{align}
	with a suitable finite constant $C>0$. Since $s_{1},s_{2} \in (0,1)$ are both arbitrary, 
	the claim follows with $s:= s_{1}s_{2}$.
\end{proof}

\begin{proof}[Proof of Corollary \ref{cor:TrClBounds2}]
	The (quasi-) norm estimates
	$\mathbb{E}\left[ |X|^{q_1} \right]^{1/q_1} \leq \mathbb{E}\left[ |X|^{q_2} \right]^{1/q_2} 
		\quad \text{ and } \quad \|A\|_{p_1}\leq \|A\|_{p_2}$
	for random variables $X$ and Schatten-class operators $A$ hold for all $0< p_2\leq p_1 < \infty$ and all 
	$0< q_1\leq q_2 < \infty$. 
	Thus, we assume without loss of generality that $p\leq 1$, $q\geq 1$ such that $k:= q/p\in\N$.

	For $f\in\mathcal F_I$ and $G\subseteq \R^d$ open we abbreviate
	$T_f:= f(H_{G})- f(H^\tau_{G})$.
	Because of the adapted triangle inequality \eqref{lem:conkav} for $p\le 1$ we estimate
	$\|T_f\|_p^p \leq \sum_{a,b\in\Z^d}\|\chi_a T_f\chi_b \|_p^p$.
	A $k$-fold application of this inequality yields
	\begin{equation}
		\label{pf:CorTrBndNew2}
		\mathbb{E}\left[ \|T_f\|_p^{q} \right] = \mathbb{E}\big[ \|T_f\|_p^{pk} \big] 
		\leq \sum_{\substack{a_1,...,a_k\in\Z^d\\b_1,...,b_k\in\Z^d}} \mathbb{E}\Big[ 
		\prod_{l=1}^k \|\chi_{a_l}T_f\chi_{b_l}\|_p^p \Big].
	\end{equation}
	Next, we apply H\"older's inequality to the expectation in \eqref{pf:CorTrBndNew2} and obtain
	\begin{equation}
		\label{pf:CorTrBndNew3}
		\mathbb{E}\left[ \|T_f\|_p^{q} \right] \leq 
		\sum_{\substack{a_1,...,a_k\in\Z^d\\b_1,...,b_k\in\Z^d}}
		\prod_{l=1}^k  \mathbb{E}\big[\|\chi_{a_l}T_f\chi_{b_l}\|_p^{pk} \big]^{1/k}.
	\end{equation}
	Now, we choose $0<s<1$. Theorem \ref{th:TrClBounds1} implies the existence of finite constants $C,\mu>0$ 
	(depending only on $p, I, s$) such that 
	$\mathbb{E}\big[\|\chi_{a_l}T_f\chi_{b_l}\|_p \big] \leq C \tau^{s/2} \NTV(f) \e^{-\mu(|a_l|+|b_l|)}$
	for all $a_{l},b_{l}\in \Z^{d}$.
	The deterministic a priori estimate from Lemma \ref{lem:aPriori} provides the existence of a finite 
	constant $C_{1}$ (depending only on $p, I, s$, but not on $\omega$) such that 
	$\|\chi_{a_l}T_f\chi_{b_l}\|_p^{pk-1}  \leq C_{1} \NTV(f)^{pk-1}$.
	We thus obtain
	\begin{equation}
		\label{pf:CorTrBndNew4}
		\mathbb{E}\big[\|\chi_{a_l}T_f\chi_{b_l}\|_p^{pk} \big] 
		\leq C C_{1} \tau^{s/2} \NTV(f)^{pk}  \e^{- \mu (|a_l|+|b_l|)}
	\end{equation}
	for all $a_{l},b_{l}\in \Z^{d}$.
	Estimating the right-hand side of \eqref{pf:CorTrBndNew3} by \eqref{pf:CorTrBndNew4} and using $q=pk$, we arrive at
	$\mathbb{E}\left[ \|T_f\|_p^{q} \right] \leq C_{2} \tau^{s/2} \NTV(f)^q$,
	where the finite constant $C_{2}$ depends on $p, I, s$.
	In particular, the constant $C_{\tau}$ in the statement vanishes algebraically as $\tau\downarrow 0$.
\end{proof}

%
\subsection{Proof of Theorem \ref{th:TrClBounds3} and Theorem \ref{Thm:prod}}

\begin{proof}[Proof of Theorem \ref{th:TrClBounds3}]
	We will follow the strategy in the proofs of Theorems \ref{th:TrClBounds1:OpNorm} and \ref{th:TrClBounds1} 
	and use the notation introduced there. As it is done there, we assume without restriction that 
	$I_{+} \ge E_{0}$ and that $0<p\leq 1$. Moreover, we restrict ourselves to the case $b\in\R^d$ 
	with $Q_b\cap G\neq\emptyset$. 

	Let $a \in\R^{d}$, $\tau\in [0,1]$, $f\in \mathcal F_I$  and define again the 
	truncation $h:= f \id_{[E_{0}-1, \infty)}$. 
	We write $\zeta_h$ for the complex measure defined as in \eqref{eq:LocHelfSjoDef1} with a cutoff 
	function $\Xi=\Xi_h$ that satisfies \ref{K1} -- \ref{K3}, where $f$ is replaced by $h$.
	Proceeding along the lines of \eqref{TrCl-done} and 
	\eqref{combes-thomas:eq10}, we obtain for any $s,s' \in (0,1)$
	\begin{align}
		\label{th:trbounds3:eq1}
		\mathbb{E} \Big[\big\| \chi_a & \big(f(H^\tau_G) - f(H^\tau_{\wtilde{G}})\big)\chi_b \big\|_p  \Big]
		 \leq C_{1} \, \|h\|_{\infty}^{1-s'}\, \mathbb{E} \left[\big\| \chi_a\big(h(H^\tau_G) - 
			h(H^\tau_{\wtilde{G}})\big)\chi_b \big\|  \right]^{s'}\notag\\
		& \leq C_{2}\, \, \|h\|_{\infty}^{1-s'} \, \E \l[  \int_{\mathbb{R}^2} \frac{\d|\zeta_h|(x,y)}{|y|^{1-s/2}}\,
			\big\|\chi_a \big(R_{x+\i y}(H^\tau_G)-R_{x+\i y}(H^\tau_{\wtilde{G}})\big)\chi_b \big\|^{s/2} \r]^{s'} \notag\\
	\end{align}
	with finite constants  $C_{1}=C_{1,p,s',I_{+}}>0$ and $C_{2}=C_{2,p,s,s',I_{+}}>0$. 
	
	\emph{Case 1.} We assume $\dist(a,\partial G)>1$ and 
	$\dist(b,\partial G)>1$. Hence, we have $Q_{b} \subset G$ in this case.
	We apply the geometric 
	resolvent inequality, see e.g.\ \cite[Lemma 2.5.2]{MR1935594}, to the operator norm in the last line
	of \eqref{th:trbounds3:eq1}. Even though it is only stated for boxes there, the key estimate, 
	\cite[Lemma 2.5.3]{MR1935594}, covers our setup.  
	Hereby we use the assumption $\dist(\partial G, \partial \wtilde{G}) >1$.
	We obtain in analogy to \cite[Eq.\ (4.9)]{DGHKM2016} 
	\begin{equation}
		\big\|\chi_a \big(R_{z}(H^\tau_{G})-R_{z}(H^\tau_{\wtilde{G}})\big)\chi_b \big\| 
		\leq C_{3} \sum_{c\in \delta G^{\#}} \|\chi_a R_{z}(H^\tau_{\wtilde{G}}) \chi_c \| 
		\|\id_{\Lambda_2(c)} R_{z}(H^\tau_{G}) \chi_b\|,
	\end{equation}
	where $\Lambda_2(c):= c+\Lambda_2$ and
	\begin{equation}
		\label{eq:Gplusraute}
 		\delta G^{\#}:=\big\{n\in\Z^d:\, \dist(n,{\partial G})\leq 1\big\}
	\end{equation}
	denotes the set of lattice points needed as centres to cover a strip of width $1$ around 
	the boundary $\partial G$ of $G \subset \R^{d}$ by unit cubes. 
	The constant $C_3$ is uniform in $z\in \C$ on each compact subset of $\C$. It is also uniform 
	in $a\in\R^d$  and $b\in G$ with $\dist(b,\partial G)> 1$. 
	Thus, the expextation in the last line of \eqref{th:trbounds3:eq1} is seen to be bounded from above by
	\begin{align}
		\label{th:trbounds3:eq2}
		  C_{4} \sum_{c\in \delta G^{\#}} \int_{\mathbb{R}^2} \frac{\d|\zeta_h|(x,y)}{|y|^{1-s/2}} \,
			\E \Big[ \big\|\chi_a R_{x+\i y}(H^\tau_{\wtilde{G}})\chi_c\big\|^s\Big]^{1/2} 
		\E\Big[ \big\| \id_{\Lambda_2(c)} R_{x+\i y}(H^\tau_{G})\chi_b \big\|^s\Big]^{1/2}
	\end{align}
	with a constant $C_{4}= C_{4,s,I_+}$. 
	We decompose the support of $\zeta_{h}$ as in \eqref{def:suppomega} 
	and treat the product of the expectations on $Z_{1}$ with the fractional-moment estimate as in \eqref{Combeseq2}
	and on $Z_{2}$ with the Combes-Thomas estimate as in and \eqref{Combeseq1} (with $H_{G}$ replaced by
	$H^\tau_{\wtilde{G}}$ in both estimates).
	The remaining integral is estimated as in \eqref{eq:zeta-int}. We then arrive at
	\begin{align}
		\label{th:trbounds3:eq3.1}
	C_{5} \sum_{c\in \delta G^{\#}} \e^{-\mu(|a-c| + |c-b| )} \big(\|h\|_{1} + \NTV(h) \big) 
		\le C_{6} \NTV(f) \e^{-(\mu/2) [\dist(a, \partial {G}) + \dist(b, \partial {G})]} \notag\\[-1ex]
	\end{align}
	as an upper bound for \eqref{th:trbounds3:eq2} with finite constants $C_{5}, C_{6}, \mu >0$, all depending only on $s$ and $I$. 
	In the last step we used \eqref{eq:1normTV}, the estimate 
	\begin{equation}
		\label{eq:partialGsum}
 		\sum_{c\in \delta G^{\#}} \e^{-\mu(|a-c| + |c-b| )} \le 
		\e^{-(\mu/2) [\dist(a, \partial {G}) + \dist(b, \partial {G})-2]} 
		\sum_{c\in \Z^{d}} \e^{-(\mu/2) (|a-c| + |c-b| )}
	\end{equation}
	and that the sum on the right-hand side of \eqref{eq:partialGsum} is 
	bounded from above uniformly in $a, b \in\R^{d}$ by, e.g., the Cauchy-Schwarz inequality.
	Now, the claim follows upon inserting \eqref{th:trbounds3:eq3.1} as an upper bound for the expectation into  the last line of \eqref{th:trbounds3:eq1} and observing
	\eqref{eq:inftynormTV}. This finishes Case 1.

	\emph{Case 2.} We assume $\dist(a,\partial G) \leq  1$ or $\dist(b,\partial G) \leq  1$. 
	Hence, we have 
	\begin{equation}
		\label{eq:ab-dist}
 		|a-b| \ge \max\big\{\dist(a,\partial G), \dist(b,\partial G)\big\} -1
	\end{equation}
	in this case.	We estimate the operator norm on the right-hand side of \eqref{th:trbounds3:eq1} 
	by the triangle inequality and each of the resulting two terms by the fractional-moment estimate
	\eqref{eq:DefFMB}, see also Remark \ref{FMBalways2}. 
	The remaining integral is again estimated by \eqref{eq:zeta-int} and \eqref{eq:1normTV}, and we obtain
	the upper bound 
$C_{7}  \e^{-\mu|a-b|} \NTV(f)$
	for the expectation in the last line of \eqref{th:trbounds3:eq1} with finite constants $C_{7}, \mu >0$ depending only on $s$ and $I$.
	Now, the claim follows upon observing
	\eqref{eq:inftynormTV} and \eqref{eq:ab-dist}.
\end{proof}

\begin{proof}[Proof of Theorem \ref{Thm:prod}]
	As in the proof of Corollary~\ref{cor:TrClBounds2} we assume without loss of generality that 
	$p\le 1$, $q \ge 1$ and we define and $k:= q/p \in\N$. Let $f\in\mathcal F_I$. 
	The argument leading to \eqref{pf:CorTrBndNew3} in the proof of Corollary~\ref{cor:TrClBounds2} 
	implies in the present context the estimate
	\begin{equation}
		\mathbb{E}\left[ \|T_{f,L}\|_p^{q} \right] \leq 
		\sum_{\substack{a_1,...,a_k\in\Z^d\\b_1,...,b_k\in\Z^d}}
		\prod_{l=1}^k  \mathbb{E}\big[\|\chi_{a_l}T_{f,L}\chi_{b_l}\|_p^{pk} \big]^{1/k},
	\end{equation}
	where $T_{f,L} := \big(f(H_{L})- f(H^\tau_{L})\big) - \big(f(H)- f(H^\tau)\big)$.
	Below, we will decompose $T_{f,L} = T_{f,L}^{(1)} - T_{f,L}^{(2)}$ in two different ways. 
	The first way is with the terms
	\begin{equation}
 		\label{eq:decomp1}
		T_{f,L}^{(1)} := f(H_{L}) - f(H) \quad \text{and} \quad T_{f,L}^{(2)} := f(H_{L}^{\tau}) - f(H^{\tau}),
	\end{equation}
	and the second is with the terms
	\begin{equation}
 		\label{eq:decomp2}
		T_{f,L}^{(1)} := f(H_{L}) - f(H_{L}^{\tau}) \quad \text{and} \quad T_{f,L}^{(2)} := f(H) - f(H^{\tau}).
	\end{equation}
	In any case, the adapted triangle inequality \eqref{lem:conkav} and Minkowski's inequality on 
	$L^{k}(\Omega,\PP)$ imply 
 	$\mathbb{E}\big[\|\chi_{a_l}T_{f,L}\chi_{b_l}\|_p^{pk} \big]^{1/k} \le \sum_{j=1}^{2} 
		\mathbb{E}\big[\|\chi_{a_l}T_{f,L}^{(j)}\chi_{b_l}\|_p^{pk} \big]^{1/k}$.
	Moreover, for either choice of the decomposition, the deterministic a priori bound from 
	Lemma~\ref{lem:aPriori} implies $\|\chi_{a_l}T_{f,L}^{(j)}\chi_{b_l}\|_{p}^{pk-1} \le C_{1} \NTV(f)^{pk-1}$ 
	with a finite constant $C_{1} = C_{1,p,I}>0$ that does not depend on $\omega$ but only on the 
	lower bound $E_{0}$  from \eqref{eq:DefE0} for the spectra. We conclude that 
	\begin{equation}
		\label{eq:TeTraceConvergence1}
 		\mathbb{E}\left[ \|T_{f,L}\|_p^{q} \right] 
		\leq C_{1} \NTV(f)^{pk-1}
		\sum_{\substack{a_1,...,a_k\in\Z^d\\b_1,...,b_k\in\Z^d}}
		\prod_{l=1}^k  \Big(\sum_{j=1}^{2}\mathbb{E}\big[\|\chi_{a_l}T_{f,L}^{(j)}\chi_{b_l}\|_p \big]^{1/k}\Big).
	\end{equation}
	Next we split the summation over each pair $(a_{l},b_{l}) \in\Z^{d} \times \Z^{d}$ into two parts: 
	the box $\Lambda^2_{L/2}:=(\Lambda_{L/2}\times\Lambda_{L/2})\cap (\Z^{{d}} \times \Z^{d}) $ and its complement
	$\Lambda^{2,c}_{L/2}:= (\ Z^{{d}} \times \Z^{d})\setminus \Lambda^{2}_{L/2}$.
	For $(a_{l},b_{l}) \in\Lambda^2_{L/2}$ we compare infinite- and finite-volume operators, that is, 
	we choose the decomposition \eqref{eq:decomp1}. Theorem \ref{th:TrClBounds3} then provides 
	finite constants $C_2 = C_{2,p,I}>0$ and $\mu_2=\mu_{2,p,I}>0$ such that
	\begin{align}
		\label{eq:TeTraceConvergence2}
  	\sum_{(a,b)\in \Lambda^2_{L/2}} \sum_{j=1}^{2} &
  		\mathbb{E}\big[\|\chi_{a_l}T_{f,L}^{(j)}\chi_{b_l}\|_p \big]^{1/k}	\notag\\
 		& \leq C_{2} \NTV(f)^{1/k} \sum_{(a,b)\in \Lambda^2_{L/2}}  
			\e^{-\mu_2[(\dist(a,\partial\Lambda_{L})+\dist(b,\partial\Lambda_{L})]/k} \notag\\
		& \leq C_{2}\NTV(f)^{1/k}\, (L+1)^{2d}  \e^{-\mu_{2} L/(2k)}.
	\end{align}

	For $(a_{l},b_{l}) \in\Lambda^{2,c}_{L/2}$ we compare unperturbed and perturbed operators, that is, 
	we choose the decomposition \eqref{eq:decomp2}. We fix $0<s<1$. Theorem~\ref{th:TrClBounds1} then 
	provides finite constants $C_3 = C_{3,p,s,I}>0$ and $\mu_3 = \mu_{3,p,s,I}>0$ such that 
	\begin{align}
		\label{eq:TeTraceConvergence3}
  	\sum_{(a,b)\in \Lambda^{2,c}_{L/2}}\sum_{j=1}^{2} 
  	\mathbb{E}\big[\|\chi_{a_l}T_{f,L}^{(j)}\chi_{b_l}\|_p \big]^{1/k}
		&	\leq 2C_{3} \NTV(f)^{1/k}\, \sum_{\substack{a\in\mathbb{Z}^d\\ b\in \Z^{d} \setminus\Lambda_{L/2}}} 
			\e^{-\mu_3(|a|+|b|)/k}\notag\\
		& \leq C_{3}' \NTV(f)^{1/k} (L+1)^{d-1} \e^{-\mu_3 L/(4k)},
	\end{align}
	where $C_{3}' = C'_{3,p,s,I}>0$ is another finite constant. We conclude from \eqref{eq:TeTraceConvergence1}, 
	\eqref{eq:TeTraceConvergence2} and \eqref{eq:TeTraceConvergence3} that
	\begin{equation}
		\label{eq:TeTraceConvergence4}
 		\mathbb{E}\left[ \|T_{f,L}\|_p^{q} \right] \leq  C \NTV(f)^{pk} \e^{-\mu L}
	\end{equation}
	with finite constants $C =C_{p,q,s,I}>0$ and  $\mu = \mu_{p,q,s,I}>0$. This proves \eqref{thm:prod:eq1}.

	The almost-sure convergence \eqref{eq:ConvergencePi} for a super-logarithmically growing 
	sequence of lengths follows from \eqref{thm:prod:eq1} with Remark~\ref{fast-convergence}.
\end{proof}


\section{Proofs of the results from Section \ref{subsec:ssf}}
\label{sec:5}

\subsection{Proof of Theorem \ref{th:SSFandTransOp}}

\begin{proof}[Proof of Theorem \ref{th:SSFandTransOp}]
	We fix $\tau \in [0,1]$.
  Let $I\subset \FMB$ be a compact interval and $E\in I$. Theorem \ref{th:SSFandTransOpii} follows from 
  Theorem \ref{th:SSFandTransOpi}  and  Theorem \ref{Thm:prod} applied to the Fermi function $f=\id_{(-\infty,E]}$.
  
	It remains to prove the left equality in Theorem \ref{th:SSFandTransOpi}, because the right equality follows 
	already from Remark \ref{rem:right-equality-clear}.
	Let $I \subset \FMB$ be again a compact interval. Then, Corollary~\ref{cor:TrClBounds2} and Fubini's Theorem imply
	$\mathbb{E} \Big[\int_{I}\d E\, \big\|T(E,H, H^\tau)\big\|_{1} \Big] <\infty$.
	Hence, we have $\|T(\,\pmb\cdot\,, H, H^\tau)\|_{1} \in L^{1}(I)$ almost surely. Thus, the left inequality 
	in Theorem \ref{th:SSFandTransOpi} follows from Lemma~\ref{lem:SSFprep2ndproof} below, and the proof is 
	complete.
\end{proof}

It remains to prove the following deterministic lemma.

\begin{lemma}
 	\label{lem:SSFprep2ndproof}
 	Let $A$ and $B$ be two self-adjoint operators in a Hilbert space $\mathcal H$ which are bounded from below. 
	We assume that 
 	$\e^{-A}- \e^{-B}\in \mathcal{S}^{1}$ and that, for some open interval $I\subset \R$, the mapping 
 	$I \ni E\mapsto \|T(E,A,B)\|_{1}$
 	is an $L^{1}(I)$-function. Then the SSF and the trace of the shift operator coincide, i.e.\
 	\begin{equation}
	 	\label{eq:SSFprep2ndproofstat}
 		\xi(E,A,B)=\Tr\big( T(E,A,B) \big)\qquad \text{ for a.e. } E\in I.
 	\end{equation}
\end{lemma}


\begin{proof}[Proof of Lemma~\ref{lem:SSFprep2ndproof}]
The assumption $\e^{-A}- \e^{-B}\in \mathcal{S}^{1}$ in the lemma is only needed for the spectral 
	shift function to be well defined according to \eqref{eq:Definition1SSF}.
	
	We show that the function $E \mapsto \Tr\big(T(E,A,B)\big)$  satisfies \eqref{eq:Definition1SSF} for
  every $f\in C^{\infty}(\mathbb{R})$ with $\supp(f')\subseteq I$ and $\lim_{\lambda\to\infty}f(\lambda)=0$. 
  
  Let $f\in C^\infty(\R)$ be such a function. By assumption we have 
	$\R \ni E\mapsto |f'(E) |\|T(E,A,B) \|_1 \in L^1(\R)$.
	Hence, $\int_\R \d E\, f'(E) \, T(E,A,B)$ is a trace-norm convergent Bochner integral. 
	Moreover, the identity
	\beq
		\label{eq:bochner}
		f(B)- f(A)=\int_\R \d E\, f'(E) \,T(E,A,B)
	\eeq
	holds, as we argue below. Since the mapping $\S^1\ni B\mapsto \Tr (B)$ is a bounded linear 
	functional on $\S^1$, it interchanges with the Bochner integral, and \eqref{eq:bochner} implies 
	the assertion
	\beq
		\Tr \big( f(B)- f(A)\big) = \int_\R \d E\,f'(E) \Tr\big(  T(E,A,B) \big).
	\eeq
	It remains to prove \eqref{eq:bochner}. We fix $\phi \in \H$ and compute, using Fubini's theorem,
	\begin{equation}
 		\<\phi , f(A) \phi\> = - \int_{\mathbb{R}}\d\mu_{\phi,\phi}^{A}(\lambda)\int_{\lambda}^{\infty}\d E\, 
			f^{\prime}(E) = -\int_{\mathbb{R}}\d E\, f^{\prime}(E) \< \phi,\id_{(-\infty,E]}(A)\phi \>,
	\end{equation}
	where we introduced the spectral measure $\mu^{A}_{\phi,\phi}:=\< \phi, \id_{\,\pmb\cdot\,}(A)\phi\>$. 
	The analogous computation for $B$ gives
	\begin{align}
		\<\phi ,\big( f(B) - f(A) \big) \phi\>  
		= \< \phi, \bigg(\int_\R \d E\, f'(E) \, T(E,A,B) \bigg) \phi \>,
	\end{align}
	where we used again the continuity of the functional $\S^1\ni B\mapsto \langle\varphi, 
	B \varphi\rangle$, which therefore interchanges with the Bochner integral. The operator equality 		\eqref{eq:bochner} follows by polarisation. 
\end{proof}

\subsection{Proof of Lemma \ref{lem:HoelderCont} and Theorem \ref{lem:SSFcont} }

\begin{proof}[Proof of Lemma \ref{lem:HoelderCont}]
	We fix $E,E'\in I$ arbitrary with $E<E'$ and define $J:=[E,E']$. 
	The estimate \eqref{TrCl-done} with $f=h=\id_{J}$ implies
	\begin{align}
		\label{eq:PfLemSSFcont1}
	  \E \big[ \big\|  T(E,H_{(L)},H_{(L)}^{\tau}) 
	  & - T(E',H_{(L)},H_{(L)}^{\tau}) \big\|_{1}\big]\nonumber\\
		& \leq C_{\theta} \sum_{a,b\in\Z^{d}} \E \big[ \big\| \chi_a \big( 
			\id_J(H_{(L)}) - \id_J(H^{\tau}_{(L)}) \big)\chi_b \big\| \big]^{\theta},
	\end{align}
	where $\theta \in (0,1)$
	and the constant $C_{\theta}$ is independent of $L$ and $\tau \in [0,1]$.
	We will apply the Helffer--Sj\"ostrand formula to the operator $\id_{J}(H_{(L)}^{(\tau)})$ 
	and need an appropriate cutoff 
	function for this purpose. Let $g\in C_c^\infty(\R)$ be a smooth indicator function such that 
	$g(x):=1$ for $x\in [-1,1]$, $g(x):=0$ for $x\in \R \backslash [-2,2]$, $\|g\|_\infty\leq 1$ and 
	$\|g'\|_\infty\leq 2$. We define the centre $J_{c}:= (E+E')/2$ of the interval $J$ and the
	cutoff function $\Xi_{J}\in C^\infty_c(\R^2)$ by 
	\beq
		\label{cutoff-specified}
		\Xi_{J}(x,y)= g\Big( \frac{x-J_{c}}{|J|}\Big)g\Big( \frac{y}{|J|}\Big), \qquad (x,y) \in \R^{2}.
	\eeq
	Let $\zeta_{\id_J}$ be the complex measure defined in \eqref{eq:LocHelfSjoDef1} with the cutoff 
	function \eqref{cutoff-specified}.
	We note that 
 	$\supp (\zeta_{\id_{J}}) \subseteq J \times  [-2|J|,2|J|]$.
	We now apply Lemma \ref{lem:LocHelfSjo} to  the difference in \eqref{eq:PfLemSSFcont1} and obtain 
	for all $a,b\in\Z^{d}$
	\begin{equation}
		\label{eq:PfLemSSFcont1b}		
		\E \left[ \big\|\chi_a \big(\id_J(H_{(L)})-\id_J(H^{\tau}_{(L)})\big)\chi_b \big\| \right]
	  \leq \frac 1 {2\pi} \int_{\mathbb{R}^2} \d|\zeta_{\id_{J}}|(x,y) \, 
	  \E\left[\|\chi_a (R_z-R^{\tau}_z)\chi_b \|\right],
	\end{equation}
	where we abbreviated $z: =x+\i y$, $R_z:=R_z(H_{(L)})$ and $R^{\tau}_z:=R_z(H^{\tau}_{(L)})$. 
	Let $p,q \in (1,\infty)$ with $1/p+1/q=1$ and let $\delta\in(0,1/p)$. 
	Then, we estimate with the help of the resolvent equation
	\begin{align}
		\label{eq:PfLemSSFcont2c}		
		\E\big[\|\chi_a (R_z &- R^{\tau}_z) \chi_b \|\big]  \notag\\
		& = \E\left[\|\chi_a (R_z-R^{\tau}_z)\chi_b \|\right]^{1/p} \,
			\E\left[\|\chi_a (R_z-R^{\tau}_z)\chi_b \|\right]^{1/q} \nonumber\\
 		& \leq \frac{2^{\delta/p}(\tau\|W\|_{\infty})^{1/q}}{|y|^{\delta+ 1/q}} \, 
			\left(\E\big[\|\chi_a R_z\chi_b \|^{1-\delta}\big]  +  
			\E\big[\|\chi_a R^{\tau}_z\chi_b \|^{1-\delta}\big]\right)^{1/p} \nonumber\\
		& \quad \times\sum_{c\in\Gamma_W} \left(\E\big[\|\chi_a R_z\chi_c\|^{1-\delta} \big] \;
			\E\big[\|\chi_c R^{\tau}_z \chi_b \|^{1-\delta}\big]\right)^{1/(2q)}, 
	\end{align}
  where $\Gamma_{W}$ was defined in \eqref{def:GammaW}. For $(x,y) \in \supp (\zeta_{\id_{J}})$ 
  we have $x \in J \subseteq I \subset \FMB$ so that each expectation can be bounded with the 
  localisation estimate from 
  Lemma~\ref{Lem:PertPersPot}, which is uniform in $\tau\in [0,1]$, $L>0$ and $J\subseteq I$. 
	Treating the resulting $c$-sum of exponentials as in \eqref{eq:expproduct} and throwing away the 
	decay in $|a-b|$, we arrive at
	\beq
		\label{eq:PfLemSSFcont3a}
		\E\left[\|\chi_a (R_z-R^{\tau}_z)\chi_b \|\right] \leq C_{1} \frac{\e^{-\mu_{1}(|a|+|b|)}}{|y|^{\delta+ 1/q}}
	\eeq
	with finite constants $C_{1},\mu_{1}>0$ that depend only on $p$, $\delta$ and $I$. 
	This estimate needs to be inserted into \eqref{eq:PfLemSSFcont1b}, 
	so the following integral is relevant
	\begin{equation}
		\label{eq:PfLemSSFcont3b}
		\int_{\R^2} \frac{\d|\zeta_{\id_{J}}|(x,y)}{|y|^{\delta+ 1/q}} \leq C_{2} |J|^{p^{-1}-\delta}.
	\end{equation}
	To obtain the bound, we used $\delta <1/p$ and applied \eqref{eq:zeta-int} with the 
	particular choice \eqref{cutoff-specified} for the cutoff function $\Xi_{J}$. The constant 
	$C_{2}>0$ depends only on $p$, $\delta$ and $I$.
	Collecting the estimates \eqref{eq:PfLemSSFcont1}, \eqref{eq:PfLemSSFcont1b}, \eqref{eq:PfLemSSFcont3a} 
	and \eqref{eq:PfLemSSFcont3b}, we conclude that 
	\beq
		\label{eq:PfLemSSFcont4}
		\E \left[ \big\| T(E,H_{(L)},H_{(L)}^{\tau}) - T(E',H_{(L)},H_{(L)}^{\tau}) \big\|_{1}\right] 
		\le  C_{3} |J|^{\theta (p^{-1} -\delta)} 
	\eeq
	with a finite constant $C_{3} >0$ that depends on $\theta$, $p$, $\delta$ and $I$, 
	but not on $J \subseteq I$, $L>0$ or $\tau\in [0,1]$. The lemma follows, because the exponent 
	$\theta (p^{-1} -\delta)$ may take on any value in $(0,1)$.
\end{proof}

Our proof of Theorem \ref{lem:SSFcont} uses the scale-free unique continuation principle 
of  \cite{MR3106507, Nakic:2015is} applied to averaged 
local traces of non-negative functions of \emph{infinite-volume} ergodic random Schr\"odinger operators. 

\begin{lemma}
 	\label{lem:UCP}
	Let $E\in\R$ and let $\Gamma \subset\R^{d}$ be a Borel set with $\Int (\Gamma) \neq\emptyset$. 
	Then there exists a finite constant $\gamma>0$ such that for every non-negative, measurable function 
	$f: \R \rightarrow [0, \infty)$ with support $\supp(f) \subseteq (-\infty,E]$ the lower bound 
	\begin{equation}
		\label{eq:UCPconsequence}
 		\E\big[ \Tr \big(\id_{\Gamma} f(H)\big)\big] \ge \gamma \int_{\R}\d E'\, \mathcal{N}'(E') f(E')
	\end{equation}
	holds. If the operator 
	$\id_{\Gamma} f(H)$ is not trace class, we define its trace to be $+\infty$.
\end{lemma}

\begin{proof}
	Even though a unique continuation principle for infinite-volume operators is known \cite{MR3519210}, 
	it seems more convenient here to use the one for finite-volume operators cited above. 
	By usual integration theory, the lemma follows if it is proven for indicator 
	functions $f$ of arbitrary Borel sets $B \subseteq (-\infty,E]$. 
	In turn, by the comparison theorem for measures, see e.g.\ \cite[Thm.\ II.5.8]{elstrodt2011measure}, 
	it is enough to prove it for semi-open intervals $I \subset (-\infty,E]$. 

	So let $I$ be such an interval.  
 	Since the probability for the endpoints of $I$ to coincide with an eigenvalue of $H$ is zero, we obtain 
	from the definition and self-averaging of the integrated density of states
	\begin{equation}
		\label{eq:IDOSlimit}
 		\int_{I} \d E' \mathcal{N}'(E') = \lim_{L\to\infty} \frac{1}{L^{d}} \, \E\big[ \Tr \big(\id_{I}(H_{L})\big)\big].
	\end{equation}
	The left-hand side of \eqref{eq:UCPconsequence} is monotone in $\Gamma$ and invariant under 
	$\Z^{d}$-translations of $\Gamma$.
	Therefore we assume without loss of generality that there exists an interior point of $\Gamma$ inside 
	the unit cube $\Lambda_{1}$ and that $\Gamma \subseteq \Lambda_{1}$.
	Then, the scale-free unique continuation principle
	for spectral projections, see \cite[Cor.]{Nakic:2015is} or \cite[Thm.\ 1.1]{MR3106507}, provides us with 
	the deterministic estimate 
	\begin{equation}
		\label{eq:UCP}
	 	\id_{I}(H_{L}) \le \frac{1}{\gamma}  \id_{I}(H_{L}) \, \id_{\Gamma_{L}} \id_{I}(H_{L})
	\end{equation}
	for $L\in\N$, where $\Gamma_{L}:= \Lambda_{L} \cap \big(\bigcup_{k \in\Z^{d}} (k + \Gamma)\big)$ 
	and the finite non-random constant $\gamma > 0$ depends only on $d$, 
	$\Gamma$, $E$ and 
 	$\mathcal{V} := \|V_{0}\|_{\infty} +  |\lambda| \, \|\sum_{k\in\Z^{d}} u_{k}\|_{\infty} < 
	\infty$,
	see Assumptions \ref{assK} and \ref{assV2}.
	Next, we insert \eqref{eq:UCP} into \eqref{eq:IDOSlimit}, exploit cyclicity of the trace and add and subtract 
	the desired $L$-independent term.  
	This yields
	\begin{align}
		\label{eq:UCPapplied}
 			\int_{I} \d E' \mathcal{N}'(E') & \le  \frac{1}{\gamma} \,  \E\big[ \Tr \big(\id_{\Gamma} \id_{I}(H)		\big)\big] \notag\\
		  & \quad + \frac{1}{\gamma} \,\liminf_{L\to\infty} 
				\left(  \frac{1}{L^{d}}  \E \big[\Tr\big( \id_{\Gamma_{L}} \id_{I}(H_{L}) \big) \big] - 
				\E \big[\Tr\big(\id_{\Gamma}\id_{I}(H) \big) \big] \right).
	\end{align}
	It remains to show that the error term in the second line of \eqref{eq:UCPapplied} 
	vanishes as $L\to\infty$. This is a question of convergence of measures. Since a given real number is not 
	an eigenvalue of $H$ with probability one, the limiting measure has no atoms, and the notion of 
	vague convergence is sufficient. Thus, we consider the Laplace transforms, and the lemma follows if 
	\begin{equation}
		\label{laplace-vanish}
 		\lim_{\substack{L\to\infty\\[.5ex] L \text{~odd}}} \frac{1}{L^{d}} \,\E \big[ \Tr \big( \id_{\Gamma_{L}} 
		(\e^{-t H_{L}} - \e^{-t H})\big) \big]   =0
	\end{equation}
	holds for every $t>0$ \cite[Thm.\ 2a in Sect.\ XIII.1]{MR0270403}. 
	Here, we reformulated the term with the $\Z^{d}$-ergodic infinite-volume Hamiltonian
	by using the the covariance relation for $H$ and the disjoint decomposition 
	$\Gamma_{L} =\bigcup_{k\in \Lambda_{L} \cap \Z^{d}} (k+ \Gamma)$, which holds for $L$ odd. 
	We analyse the semigroups in 
	the Feynman-Kac representation \cite{MR2105995}, see also 
	\cite[Sect.\ 6]{MR1756112},
	and obtain 
	\begin{align}
		\label{BB-rep}
		\bigg|\frac{1}{L^{d}} 
		\,\E \big[ \Tr \big( \id_{\Gamma_{L}} 
		(\e^{-t H_{L}} - \e^{-t H})\big) \big] \bigg| 
			\le \frac{\e^{t\mathcal{V}}}{(2\pi t)^{d/2}} \int_{\Lambda_{L}} \frac{\d x}{L^{d}} \; 
			\mathbf{E}_{0,x}^{t,x}( 1- \chi^{t}_{\Lambda_{L}} ),
	\end{align}
 	where 
	$\mathbf{E}_{0,x}^{t,x}$ denotes the Brownian-bridge expectation over continuous closed paths starting 
	at time $0$ at $x$ and returning to $x$ at time $t$, and $\chi^{t}_{\Lambda_{L}}$ is the indicator 
	function of the set of Brownian-bridge paths which stay inside $\Lambda_{L}$ for all times up to $t$. 
	But the integral on the right-hand side of \eqref{BB-rep} vanishes as $L \to\infty$ according to 
	\cite[p.\ 341]{artRSO1989Kir}.
\end{proof}

\noindent
We are now ready for the

\begin{proof}[Proof of Theorem \ref{lem:SSFcont}]
	Throughout the proof we abbreviate the SSF by 
	$\xi_{(L)} := \xi(\,\pmb\cdot\,,H_{(L)},H_{(L)}^{\tau})$ and choose representatives which coincide with 
	the trace of the corresponding spectral shift operator on $\FMB$. Since $\xi_{(L)}$ depends monotonously on the 
	perturbation and $W\geq C u_0$, we assume without loss of generality that $W=C u_0$. 

	By Lemma \ref{lem:HoelderCont}, the function $E \mapsto \E[\xi(E)]$ is H\"older continuous on compact intervals in 
	$\FMB$ for any $\alpha\in (0,1)$. Let $E\in \FMB$ and $\varepsilon_{0} >0$ such that 
	$I_{\varepsilon_{0}} := [E-\varepsilon_{0}, E+ \varepsilon_{0}] \subset \FMB$. Consider any 
	$\varepsilon \in (0,\varepsilon_{0}]$. Then
	\begin{equation}
		\label{eq:PfPosSSFcont1}
		\E \left[ \xi(E) \right] \geq \frac{1}{2\eps} \int_{I_{\varepsilon}}\d E'\, \E \left[ \xi(E') \right]  
		- C_{1} \eps^{\alpha},
	\end{equation}
	where the H\"older constant depends on $\alpha$, $E$ and $\varepsilon_{0}$, but not on $\varepsilon$. 
	We denote by $\E_{\neq 0}[\,\pmb\cdot\,]$ the averaging with respect to all random variables but $\omega_0$ 
	and infer from the Birman-Solomyak formula \cite{MR0315482} 
	\begin{align}
		\label{eq:PfPosSSFcont2}
 		\int_{I_{\varepsilon}} \d E'\,\E \left[ \xi(E') \right]  &= \int_{0}^1 \d s\, 
		\E \left[ \Tr\left( Cu_0 \,\id_{I_\eps}( H+s C u_0 )\right) \right]\nonumber  \\
		& = \int_{0}^{C} \d s \int_{0}^1 \d\omega_0\, \rho(\omega_0) \, \E_{\neq 0} 
			\big[ \Tr\big( u_0 \id_{I_\eps}(H+su_0) \big) \big].
	\end{align}
	We fix a parameter $s_{0} \in (0,\min\{ 1,C\}]$ to be determined later. 
	Performing the change of variables $ \omega_{0} \mapsto \omega_0 +s$ 
	and restricting first the $s$-integration to $[0,s_0]$ and then the $\omega_{0}$-integration to $[s_0,1]$, 
	we obtain the estimate
	\begin{align}
		\label{eq:PfPosSSFcont2b}		
		\int_{I_{\varepsilon}} \d E'\,\E \left[ \xi(E') \right] & \geq 
			s_0 \rho_- \int_{s_0}^{1} \d\omega_0\, \E_{\neq 0} \big[ \Tr\big( u_0 \id_{I_\eps}({H})\big) \big] \nonumber\\
		& \geq s_0 \rho_- \left( \frac{1}{\rho_+} \E \left[ \Tr\left( u_0 \id_{I_\eps}( H) \right) \right] 
			- \int_{0}^{s_0} \d\omega_0\, \E_{\neq 0} \left[ \Tr\left( u_0 \id_{I_\eps}( H) \right) \right] \right)\nonumber\\
		& \ge  s_0 \rho_- \left(\frac{\gamma \mathcal{N}(I_{\varepsilon})}{\rho_+} - J\right). 
	\end{align}
	Here, the last inequality follows from the right estimate in \eqref{bounds:u0} together with the unique 
	continuation principle stated in Lemma~\ref{lem:UCP}. Furthermore, we introduced 
	$\rho_+ := \esssup_{x\in[0,1]} \rho(x)<\infty$ and 
	$
		J := \int_{0}^{s_0} \d\omega_0\, \E_{\neq 0} \left[ \Tr\left( u_0 \id_{I_\eps}( H) \right) \right]. 
	$

	To estimate $J$ we first exchange the operator $H$ by its finite-volume restriction $H_L$. The error 
	arising from this modification of the operator can be bounded by Theorem \ref{th:TrClBounds3}, 
	which yields finite constants $C_2,\mu>0$ such that
	$\E \left[\left\vert  \Tr\left( u_0 \id_{I_\eps}( H) \right) - 
		\Tr\left( u_0 \id_{I_\eps}( H_L) \right)\right\vert \right]  
		\leq   C_2  \e^{-\mu L}$
	for every $L>0$. This estimate implies the bound 
	\begin{equation}
		\label{eq:PfPosSSFcont3a}
		J \leq s_0 \sup_{\omega_0\in [0,s_0]}  \E_{\neq 0}\left[ \Tr\left( u_0 \id_{I_\eps}(H_L) \right) \right] 
		+ \frac{C_{2}}{\rho_{-}} \, \e^{-\mu L}
	\end{equation}
	for every $L>0$. Let $e_1:=(1,0,...,0) \in\R^{d}$ be the unit vector along the first coordinate axis. 
	We define the subset of lattice points $A_L:= \big( e_{1}+ (3\Z)^d \big) \cap \Lambda_L$ and introduce the 
	random background operator
	$\wtilde{H}_L :=  H_{0,L} +  \lambda \sum_{k\in\Z^d\setminus A_L} \omega_k u_k$,
	where $H_{0,L}$ is the Dirichlet restriction of $H_{0}$ to $\Lambda_{L}$. 
	However, for any fixed realisation of coupling constants $\{\omega_k\}_{k\in \Z^{d}\setminus A_{L}}$, 
	we view $\wtilde{H}_L$ 
	as a non-random operator with a non-periodic background potential that is bounded uniformly in
	$\{\omega_k\}_{k\in \Z^{d}\setminus A_{L}}$. Now, the Dirichlet restriction $H_{L}$ takes the form 
	$H_L=\wtilde{H}_L + \lambda \sum_{k\in A_L} \omega_{k} u_{k}$
	and, after scaling by a factor $1/3$ and (if required) introducing an energy shift, it 
	constitutes a \emph{crooked Anderson Hamiltonian} in the sense of \cite{MR3106507}. 
	We apply the Wegner estimate \cite[Thm.\ 1.4]{MR3106507} and obtain a finite constant $C_3 >0$ such that
	\begin{align}
		\label{eq:PfPosSSFcont8}
		\sup_{\omega_k\in [0,1], \, k\in \Z^d\setminus A_L} \E_{A_L}\left[  \Tr\left( u_0 \id_{I_\eps}( H_L )\right) \right] 
		& \leq \|u_{0}\|_{\infty} \sup_{\omega_k\in [0,1],\, k\in \Z^d\setminus A_L} 
			\E_{A_L}\left[  \Tr\left( \id_{I_\eps}( H_L )\right) \right]	\nonumber\\
		& \leq C_3 2\eps  L^d, 
	\end{align}
	where $\E_{A_L}[ \,\pmb\cdot\, ]$ denotes the average over the couplings $\{\omega_k\}_{k\in A_L}$. The estimate  
	\eqref{eq:PfPosSSFcont8} holds for all $\varepsilon \in (0,\varepsilon_{1}] $ and all length-scales $L  \ge L_{1}$ 
	such that $L/3 \in \N$ is odd, and $\varepsilon_{1}$ and $L_{1}$ depend only on model parameters. We insert 
	\eqref{eq:PfPosSSFcont8} into \eqref{eq:PfPosSSFcont3a} and obtain 
	\begin{equation}
		\label{eq:PfPosSSFcont9}
		J \leq C_3  2\eps s_0 L^d + \frac{C_2}{\rho_-} \e^{-\mu L}.
	\end{equation}
	Combining \eqref{eq:PfPosSSFcont1}, \eqref{eq:PfPosSSFcont2b} and \eqref{eq:PfPosSSFcont9}, we conclude
	\begin{align}
		\label{eq:PfPosSSFcont10}
		\E \left[ \xi(E) \right] & \geq  \rho_- s_{0} \left( \frac{\gamma}{\rho_{+}}	
		\frac{\mathcal{N}(I_{\varepsilon})}{2 \varepsilon}  -  C_3  s_{0} L^d 
		 - \frac{C_2}{\rho_-}  \, \frac{\e^{-\mu L}}{2\varepsilon}  \right)  - C_1 \varepsilon^{\alpha}
	\end{align}
	for all $\varepsilon \in (0, \min\{\varepsilon_{0},\varepsilon_{1}\}]$,
	all $s_{0} \in (0,\min\{ 1,C\}]$ 
	and all $L \ge L_{1}$ such that $L/3 \in \N$ is odd.  

	Now, suppose that $E$ is a Lebesgue point of the integrated density of states $\mathcal{N}$, 
	which is the case Lebesgue-almost everywhere.  Then we have 
	$\mathcal{N}(I_{\eps})/2\eps \to \mathcal{N}'(E)$ as $\varepsilon\downarrow 0$. 
	Suppose also that $\mathcal{N}'(E)>0$. Then, there exists 
	$\varepsilon_{2} \in (0, \min\{\varepsilon_{0},\varepsilon_{1}\}]$ such that
	$\mathcal{N}(I_{\eps})/2\eps \ge \mathcal{N}'(E)/2$ for every $\varepsilon \in (0,\varepsilon_{2}]$. 
	Finally, we choose $L \ge L_{1}$ with $L/3 \in \N$ odd so large that $\eps := L^{-4d/\alpha} \le \varepsilon_{2}$ 
	and $s_0:=L^{-2d} \le \min\{1,C\}$. In this case \eqref{eq:PfPosSSFcont10} yields 
	\begin{align}
		\label{ssf-lb-final}
		\E \left[ \xi(E) \right] & \geq \frac{\gamma \rho_{-}}{2\rho_{+}} \, L^{-2d}  
		\Big( 	\mathcal{N}'(E)  -  O(L^{-d})  \Big).
	\end{align}
	The right-hand side of \eqref{ssf-lb-final} is strictly positive by possibly enlarging $L$ even further. 
\end{proof}


\section{Proofs of the results from Section \ref{subsec:aoc}}
\label{sec:6}

A crucial observation to prove Theorem \ref{th:IntMainRes} is the following Fredholm determinant representation of the ground-state overlap. 

\begin{lemma}
	\label{Fredholm:Repr}
	Let $N\in\N$, $L>0$ and define the orthogonal projections
	\begin{equation} 
		\label{eq:projections_till_N}
  	P:= P_{N,L} := \sum_{j=1}^N |\varphi_j^L\>\<\varphi_j^L|
	  \quad\text{and}\quad
  	Q:= Q_{N,L} := \sum_{k=1}^{N} \big|\psi_k^L\big\> \big\<\psi_k^L\big|,
	\end{equation}
	which we have written down in Dirac notation. Then
	\begin{align}
		\label{fredholm1}
		S_{N,L} &= \det\big( \id- (P - Q)^2\big)^{1/4}  \notag \\
    & =  \det\big( \id- P(\id - Q) P \big)^{1/2} 
    =  \det\big( \id- (\id - P) Q (\id- P) \big)^{1/2}.
	\end{align}    
\end{lemma}

\begin{proof}
  Consider the $N\times N$-matrix $M:= \big(\<\varphi^L_j, \psi^L_k\>\big)_{1\le j,k\le N}$. 
  Then, the matrix entries of $MM^*$, respectively $M^*M$, read 
  	$(MM^*)_{jl} 
  		= \< \varphi_j^L,P Q P\varphi_l^L\>$, respectively 
  	$(M^*M)_{jl} 
			= \< \psi_j^L, Q P Q\psi_l^L\>$,
  for $1\le j,l \le N$.
	Thus, $S_{N,L}$ can be written as
  \begin{equation}
    \label{Fredholm:eq4}
    S_{N,L}^2 = \det ( M M^*)
	      = \det \big( P Q P \big|_{\ran P}\big) 
	      = \det \big( \id - P Q^{c} P\big) ,
  \end{equation}
  where $Q^{c} := \id -Q$. Likewise, we have
  \begin{align}
    \label{Fredholm:eq5}
    S_{N,L}^2 = \det (M^*M)
	        =  \det \big( Q P Q  \big|_{\ran Q}\big) 
	      & =  \det\big(\id - QP^{c} Q \big) 
	       =  \det\big(\id - P^{c} Q P^{c} \big).
  \end{align}
  Here, the last equality follows from the fact that the non-zero singular values of $QP^{c}$ 
  coincide with the non-zero singular values of its adjoint $P^{c}Q$.
	The determinants in \eqref{Fredholm:eq4} and \eqref{Fredholm:eq5} are well-defined Fredholm 
	determinants because $PQ^{c} P$ and $QP^{c} Q$ are of finite rank. Multiplying the expressions in \eqref{Fredholm:eq4} and \eqref{Fredholm:eq5} yields
  \begin{align}
    \label{Fredholm:eq1}
    S_{N,L}^4 &= \det \big( \id - P Q^{c} P\big)  \det\big(\id - P^{c} Q P^{c} \big)
	   = \det \big( \id - PQ^{c} P -  P^{c} Q P^{c} \big).
  \end{align}
  Now, the lemma follows from the identity 
  \begin{equation}
 		\label{eq:square-products}
		(P-Q)^{2} = PQ^{c}P + P^{c}QP^{c}.
	\end{equation} 
\end{proof}

The remainder of this section is devoted to the proof of Theorem \ref{th:IntMainRes}. 
For $L>0$ and $E\in \R$ we abbreviate the index and the shift operator of $H_{(L)}$ and $H_{(L)}^\tau$ by 
 $\theta_{(L)}(E):=\theta(E,H_{(L)},H^\tau_{(L)})$ and 
 $T_{(L)}(E) := T(E,H_{(L)},H^\tau_{(L)})$,
respectively.

Note that $P_{N_{L}(E),L} = \id_{(-\infty, E]}(H_{L})$. If $\theta_{(L)}(E)=0$ holds, 
then we also have the identity $Q_{N_{L}(E),L} = \id_{(-\infty, E]}(H_{L}^{\tau})$. In this case we obtain
 \begin{equation}
 	\label{eq:Srep:xi0}
 	S_{L}(E) = \det\big(\id - T(E, H_{L},H_{L}^{\tau})^{2}\big)^{1/4}.
\end{equation}

\begin{proof}[Proof of Theorem \ref{th:IntMainRes}]
	We address the different parts of the theorem in the order \itemref{th:IntMainResiprime}, 
	\itemref{th:IntMainResii}, \itemref{th:IntMainResiii}, \itemref{th:IntMainResi}.
	
	\textsc{Part \itemref{th:IntMainResiprime}.} \quad 
	Let $E\in \FMB$. Then Corollary \ref{cor:TrClBounds2} implies that $T(E)\in\mathcal{S}^2$ almost surely, 
	and we have $S(E)= \det\big(\id- T(E)^2\big)^{1/4}$ almost surely. 
	Clearly, if $1$ is an eigenvalue of $T(E)^{2}$, then $S(E)=0$. Conversely, suppose that $1$ is not an 
	eigenvalue of $T(E)^{2}$. Observing \eqref{eq:T<1}, we denote by 
	$1 > \|T(E)\|^{2} = b_{1} \ge b_{2} \ge \ldots \ge 0$ 
	the non-increasingly ordered sequence of eigenvalues of $T(E)^{2}$. Then we obtain
	\begin{align}
		\label{eq:S-expand}
 		S(E)^{4} &= \exp\bigg\{ \sum_{n\in\N} \ln (1-b_{n})\bigg\}
			= \exp\bigg\{ -\sum_{n\in\N} \sum_{k\in\N} \frac{b_{n}^{k}}{k}\bigg\} \notag\\
		& \ge \exp\bigg\{ - \sum_{n\in\N} b_{n}\sum_{k\in\N} b_{1}^{k-1}\bigg\} 
			= \exp\bigg\{ - \frac{\|T(E)\|_{2}^{2}}{1- \|T(E)\|^{2}}\bigg\}
	 		>0. 
	\end{align}

	\textsc{Part \itemref{th:IntMainResii}.} \quad 
	We fix a compact interval $I \subset\FMB$, $E\in I$, a constant $c>0$ and define the event 
	$\Omega_{0} := \{ \|T(E)\|_{1} <1\}$. Markov's inequality implies
	\begin{align}
		\label{eq:FinalAbsOfAndOrth11}
		\mathbb{E}\big[ S(E) \big] 
		& \geq \e^{-c/4}\ \mathbb{P}\left[ S(E)^4\geq \e^{-c} \right]
			\ge  \e^{-c/4}\ \mathbb{P}\big[ \Omega_{0} \cap \{ S(E)^4\geq \e^{-c}\}  \big] \notag\\
		& \ge \e^{-c/4}\ \mathbb{P}\bigg[ \Omega_{0} \cap \bigg\{ \exp \Big(- \frac{\|T(E)\|_{2}^{2}}{1- \|T(E)\|^{2}}
			\Big) \geq \e^{-c} \bigg\} \bigg],
	\end{align}
	where we used \eqref{eq:S-expand} to obtain the last inequality. 
	Since $\|T(E)\| \le \|T(E)\|_{1}< 1$ on $\Omega_{0}$ and, thus, 
	$\|T(E)\|_{2}^{2} \le \|T(E)\| \|T(E)\|_{1} < \|T(E)\|_{1}$, we conclude that the fraction inside 
	the exponent in the last line of \eqref{eq:FinalAbsOfAndOrth11} is bounded from above by 
	$\|T(E)\|_{1}/(1- \|T(E)\|_{1})$. Therefore we get
	\beq
		\mathbb{E}\big[ S(E) \big] \ge  \e^{-c/4}\ \mathbb{P}\Big[  \|T(E)\|_1 \leq \frac{c}{c+1} \Big].
	\eeq
	We apply Markov's inequality again and obtain
	\begin{equation}
		\label{eq:FinalAbsOfAndOrth13}
		\mathbb{P}\Big[ \|T(E)\|_1 \leq \frac{c}{c+1} \Big] = 1-\mathbb{P}\Big[  \|T(E)\|_1  > \frac{c}{c+1} \Big]	
		\geq 1- \frac{c+1}{c} \,\mathbb{E}\big[ \|T(E)\|_1 \big]
	\end{equation}
	so that 
	\begin{equation}
		\label{eq:FinalAbsAndOrth14}
		\mathbb{E}\big[ S(E) \big] \geq \e^{-c/4} \Big( 1-  \frac{c+1}{c} \,\mathbb{E}\big[ \|T(E)\|_1 \big] \Big).
	\end{equation}
	Corollary \ref{cor:TrClBounds2} gives $\sup_{E \in I} \E\big[\|T(E)\|_1\big]\to 0$ as $\tau\to 0$. 
	Hence, the claim follows by letting $c\downarrow 0$.

	\textsc{Part \itemref{th:IntMainResiii}.} \quad 
	We restrict ourselves to $ W\geq 0$, the proof for $ W\leq 0$ follows along the same lines. 
	Since $\theta(E,H,H^\tau) = \xi(E,H,H^\tau)$ for a.e.\ $E\in\FMB$ almost surely by 
	Theorem \ref{th:SSFandTransOpi}, the equivalence \eqref{AOCequiIndex} implies
	\beq\label{eq:equivprob}
		\PP\big[S(E)=0\big] = \PP \big[ \theta(E,H,H^\tau) >0\big] = \PP\big[ \xi(E,H,H^\tau) >0 \big]
	\eeq
	for a.e.\ $E\in\FMB$. Hence, the claim follows from Theorem \ref{lem:SSFcont}.

	\textsc{Part \itemref{th:IntMainResi}.} \quad 
	We fix $E \in\FMB \cap \Int(\Sigma)$.
	The almost-sure statement \eqref{eq:ResultNoAocStatPointw} implies \eqref{eq:ResultNoAocStat1} 
	by dominated convergence and a subsubsequence argument applied to the $\mathbb{R}$-valued 
	sequence $\left(\mathbb{E}\left[S_L(E)\right]\right)_{L>0}$. 
	Thus it remains to prove \eqref{eq:ResultNoAocStatPointw}. 

	Let $(L_n)_{n\in\mathbb{N}}$ be a sequence with $L_{n}/\ln n \to\infty$ as $n\to\infty$. 
	Since the quantities $\theta_{(L_{n})}(E)$ are all integer-valued, \eqref{eq:SSFandIndex} implies 
	that there is a random variable $n_{0}:= n_{0}(E):\Omega \rightarrow\N$ such that almost surely
	\begin{equation}
		\label{theta-const}
  	\theta(E)=\theta_{L_n}(E) \qquad \text{for all } n\geq n_0.
	\end{equation}

 	\emph{Case 1: $\theta(E)=0$.} \quad
	Let $n \ge n_{0}$. Because of \eqref{theta-const}, \eqref{eq:trShift-index} and 
	\eqref{eq:xi-T} for $L=L_{n}$ we have the representation \eqref{eq:Srep:xi0} for the finite-volume overlap
	in this case. The continuity of the Fredholm determinant with respect to the trace norm 
	\cite[Thm.\ 3.4]{simon2005trace} implies
	\begin{align}
		\label{eq:ProofIndexZero}
		\big|S_{L_{n}}(E)^{4} - S(E)^{4}\big| & =
		\big|\det\big( \id - T_{L_n}(E)^2\big) - \det\big(\id - T(E)^2\big)\big| \notag \\
		& \leq 2 \| T_{L_n}(E) - T(E) \|_1\exp\big( \|T_{L_n}(E)\|_2^2+ \|T(E)\|_2^2 +1\big).
	\end{align}
	Theorem \ref{Thm:prod} now yields the almost-sure convergence
	\begin{equation}
		\lim_{n\to\infty} \| T_{L_n}(E) - T(E) \|_1\exp\big( \|T_{L_n}(E)\|_2^2+ \|T(E)\|_2^2 +1\big) =  0.
	\end{equation}
	This and \eqref{eq:ProofIndexZero} imply the desired convergence \eqref{eq:ResultNoAocStatPointw} 
	in the case $\theta(E)=0$. 

	\emph{Case 2: $\theta(E) \neq 0$.} \quad
	We assume without loss of generality that $\theta(E)>0$. The other case follows along the same lines. 
	Because of the equivalence \eqref{AOCequiIndex} we have $S(E)=0$ in this case, which, in turn, is equivalent 
	$1\in \sigma\big(T(E)^2\big)$ by \eqref{eq:ResultNoAocStat2}.
	We abbreviate $P_{n}:= P_{N_{L_n}(E), L_n}$ and $Q_{n}:= Q_{N_{L_n}(E), L_n}$. 
	Let us assume the strong convergence 
	\beq
		\label{eq:strongconv}
		\mathrm{s}\,\text{-}\!\!\lim_{n\rightarrow\infty} ( P_n -  Q_n)^2 =  T(E)^2
		\qquad\text{almost surely}
	\eeq
	for the time being. Since $1\in \sigma\big(T(E)^2\big)$, it implies almost surely the existence of a sequence of 
	eigenvalues $(\alpha_n)_{n\in\N}$ with $\alpha_n\in \sigma \big( (P_n -Q_{n})^{2}\big)$ and 
	$\alpha_n\to 1 $ as $n\to\infty$ \cite[Thm.\ VIII.24(a)]{reed1980methods1}. Moreover, 
	$0 \leq (P_n -Q_{n})^2 \leq 1$, and we conclude 
	(note that $N_{L_{n}}(E) \neq 0$ for $n$ large enough because $E \in \Int(\Sigma)$, 
	whence the first line in \eqref{eq:DefOverlap} applies)
	\beq
  	S_{L_n}(E)^{4} = \det\big( \id - (P_n -  Q_n)^2\big) \leq 1- \alpha_n \rightarrow 0
	\eeq
	as $n \to\infty$ almost surely. This is the assertion in the case $\theta(E) \neq 0$. 

	Thus, it remains to prove \eqref{eq:strongconv}. For this it suffices to show strong convergence 
	$P_n - Q_{n} \to T(E)$ as 
	$n\to\infty$. Let $\eta\in L^2(\R^d)$. We add and subtract the term $T_{L_n}(E)$ and estimate
	\begin{equation}
		\label{eq:pfAOC1} 
		\big\| \big( P_n -   Q_n - T(E) \big) \eta \big\| 
		\leq \big\| \big(\id_{(-\infty,E]}(H^\tau_{L_n}) - Q_n\big)\eta\big\|
		+	\big\|\big( T_{L_n}(E) - T(E)\big)\eta \big\|,
	\end{equation}
  where we used that $ P_{n} = \id_{(-\infty,E]}(H_{L_n})$ by definition of the particle number
  \eqref{eq:DefParticleNumberSeq}.
	The second term on the right-hand side of \eqref{eq:pfAOC1} converges to $0$ as $n\to\infty$ almost surely
	by Theorem~\ref{Thm:prod}.
  As to the first term on the right-hand side we recall from \eqref{theta-const}, 
	\eqref{eq:trShift-index} and $\Tr P_{n} = \Tr Q_{n}$ that 
 	$\theta(E) = \theta_{L_{n}}(E) = \Tr\big( T_{L_{n}}(E) \big) 
		= \Tr\big( Q_n - \id_{(-\infty,E]}(H_{L_{n}}^{\tau}) \big)$
	almost surely for $n\ge n_{0}$. Since we assumed that $\theta(E)>0$ and $Q_n$ is the orthogonal projection on the eigenspaces of all eigenvalues up to the $N_{L_n}(E)$th eigenvalue of $H^\tau_{L_n}$, we obtain that
	$\mu_{N_{L_n}(E)-\theta(E)}\leq E\leq \mu_{N_{L_n}(E)-\theta(E)+1} \leq\mu_{N_{L_n}(E)}$
	and that
  \beq\label{eq:pfAOC1c}
  	Q_n - \id_{(-\infty,E]}(H^\tau_{L_n}) =
  		\displaystyle\sum_{k= 0}^{\theta(E)-1} |\psi^{L_{n}}_{N_{L_n}(E)-k}\> \<\psi^{L_{n}}_{N_{L_n}(E)-k}|
   \eeq
  hold almost surely for $n \ge n_{0}$. 
  Since $E \in\Int(\Sigma) $ and $\Sigma = \sigma_{\text{ess}}(H^{\tau})$, we obtain for any $\eps>0$ by Fatou's lemma and strong resolvent convergence of $H_{L_{n}}^{\tau}$ to $H^{\tau}$ that
	$\lim_{n\to\infty} \Tr \big(\id_{[E,E+\eps)}(H_{L_{n}}^{\tau}) \big)
		\ge \Tr\big(\id_{[E,E+\eps)}(H^{\tau}) \big)= \infty$
	almost surely.
 Since $\theta(E)$ is a finite number,  this implies that for any $\eps>0$ there is an $n_1>0$ such that
 \beq
 \displaystyle\sum_{k= 0}^{\theta(E)-1} |\psi^{L_{n}}_{N_{L_n}(E)-k}\> \<\psi^{L_{n}}_{N_{L_n}(E)-k}|\leq 
 \id_{[E,E+\eps)}(H^{\tau})
 \eeq
 for all $n\geq n_1$ almost surely. 
 This yields for any fixed $\varepsilon>0$ the bound
  \beq
  	\label{eq:pfaoc2}
  	\limsup_{n\to\infty} \big\| \big( Q_{n} - \id_{(-\infty,E]}(H^\tau_{L_n}) \big)\eta\big\|
  	\leq \limsup_{n\to\infty} \big\| \id_{[E,E+\eps)} (H^\tau_{L_n})\eta \big\|
  \eeq
  almost surely.  
	Finally, we fix a null sequence $(\varepsilon_{l})_{l\in\mathbb{N}} \subset (0,\infty)$. 
	Lemma \ref{lm:multiEigenvalues} implies that $E,E+\varepsilon_{l}$, $l\in\mathbb{N}$, are no eigenvalues of 
	$H^\tau$ almost surely.  Therefore we have strong convergence $\id_{[E,E+\varepsilon_l)}(H^\tau_L)
	\to \id_{[E,E+\varepsilon_l)}(H^\tau)$ as $L\to\infty$ almost surely \cite[Thm.\ VIII.24(b)]{reed1980methods1}. 
	This and \eqref{eq:pfaoc2} lead to
	\begin{align}
  	\label{eq:pfaoc4}
		\limsup_{n\to\infty} \big\| \big( Q_{n} - \id_{(-\infty,E]}(H^\tau_{L_n}) \big)\eta\big\|
  	 \leq  \limsup_{l\to\infty} \big\| \id_{[E,E+\eps_l)} (H^\tau)\eta \big\| 
  	 = \big\| \id_{\{E\}} (H^\tau)\eta \big\|=0
  \end{align}
	almost surely, where the last equality follows again from $E$ being almost surely not an eigenvalue of $H^\tau$. 
	This proves \eqref{eq:strongconv}.
\end{proof}


\begin{appendix}

\section{Stability of fractional-moment bounds}
\label{app:FMB}

Here, we prove that the fractional-moment bounds given in Definition \ref{DefFMB} are stable under 
local perturbations. Among others, we show that $\FMB(H)=\FMB(H +W)$ for per\-turbations 
$W\in L_c^\infty(\R^d)$.

\begin{lemma}
	\label{Lem:PertPersPot}
	Let  $I\subset\FMB$ be a compact interval. Then for any fixed $0<s<1$ there exist constants $C,\mu>0$ 
	such that for all $G\subseteq\Rn$ open and $a,b \in \Rn$
	\begin{equation}
		\label{eq:PertPersPotStat}
		\sup_{\tau\in [0,1]}\sup_{E\in I,\eta\neq 0} \mathbb{E}\left[\|\chi_a R_{E+i\eta}(H^\tau_G) 
		\chi_b\|^s \right] \leq C \e^{-\mu|a-b|}
	\end{equation}
	holds. In particular, $\FMB(H)=\FMB(H^\tau)$ for every $\tau \in [0,1]$.
\end{lemma}

\begin{proof}
	Let $z:=E+i\eta$ with $E\in I$ and $\eta\neq 0$. For the moment, we also fix $0<s<1/2$.  
	The resolvent equation
	$R_{z}(H_{G}^{\tau}) = R_{z}(H_{G}) -\tau R_{z}(H_{G}) W R_{z}(H_{G}^{\tau})$
	yields the upper bound
	\begin{align}
		\mathbb{E} \left[ \Vert\chi_{a} R_{z}(H_{G}^{\tau})\chi_{b}\Vert^s \right]  
		& \le \mathbb{E} \left[ \Vert\chi_{a}R_{z}(H_{G})\chi_{b}\Vert^s \right]
			+ \tau^{s} \mathbb{E} \left[\Vert\chi_{a} R_{z}(H_{G}) W R_{z}(H_{G}^{\tau})\chi_{b}\Vert^s\right] \notag\\
		& =: I_{1}+ \tau^{s} I_{2}.
	\end{align}
	While the first term $I_{1}$ is independent of $\tau$ and can directly be estimated by 
	the fractional-moment bound \eqref{eq:DefFMB}
	for the unperturbed operator $H_{G}$, we estimate $I_2$ as
	\begin{align}
		I_{2} & \leq \|W\|_\infty^{s} \sum_{c\in\Gamma_W} \mathbb{E} \left[ \Vert\chi_{a}R_{z}(H_{G})\chi_{c}\Vert^s 
			\Vert\chi_{c} R_{z}(H_{G}^{\tau}) \chi_{b}\Vert^s \right] \notag\\
		& \leq  \|W\|_\infty^{s} \sum_{c\in\Gamma_W}  
			\mathbb{E} \left[ \Vert\chi_{a} R_{z}(H_{G}) \chi_{c}\Vert^{2s}\right]^{1/2} 
			\mathbb{E} \left[\Vert\chi_{c} R_{z}(H_{G}^{\tau}) \chi_{b}\Vert^{2s} \right]^{1/2} ,
	\end{align}
	where $\Gamma_{W}$ is defined in \eqref{def:GammaW}. Now, we recall Remark~\ref{FMBalways} and 
	estimate the first expectation
	with the fractional-moment bound for the unperturbed operator 
	$H_{G}$ and the exponent $2s <1$. And we estimate the second expectation with the a priori estimate
	\eqref{eq:AppAPriori} which is uniform in $\tau \in [0,1]$. This yields			
	\begin{equation}
		\label{AppLm1eq1}
		I_{2} \leq C_1 \sum_{c\in\Gamma_W}  \e^{-\mu_1|c-a|} \leq C_2 \e^{-\mu |a|}
	\end{equation}
	with finite constants $C_{1}, C_{2}>0$ that are independent of $G \subseteq\R^{d}$ open, $a \in\R^{d}$ 
	and $\tau \in [0,1]$. For the second inequality we used that $\Gamma_{W}$ is finite due to the 
	compact support of $W$. Since $R_{z}(H_{G}) W R_{z}(H_{G}^{\tau}) = R_{z}(H_{G}^{\tau}) W R_{z}(H_{G})$, 
	we obtain along the same lines that $I_2\leq C_2 \e^{-\mu |b|}$. Multiplying this inequality with 
	\eqref{AppLm1eq1}, we infer
	\beq
		\label{AppLm1eq2}
		I_2\leq C_2 \e^{-(\mu /2)(|a|+|b|)} \le C_{2} \e^{-(\mu /2)(|a - b|)}.
	\eeq
	Together with the bound for $I_{1}$, which decays exponentially in $|a-b|$, we obtain the assertion 
	for $s<1/2$. For general $0<s<1$, the result then follows from Remark~\ref{FMBalways} where no 
	additional $\tau$-dependence is generated.
\end{proof}

We finish this appendix with another simple stability result. It shows that
the probability for a given energy to be an 
eigenvalue is zero. This property does not rely on ergodicity, it follows already from the existence of the a priori bounds 
\eqref{eq:AppAPriori}.

\begin{lemma}\label{lm:multiEigenvalues}
	Let $E\in\R$ and $\tau\in [0,1]$ be given. Then $E$ is not an eigenvalue of $H^\tau$ almost surely.  
\end{lemma}

\begin{proof}
Since, for $x\in\R$, $-i\eps (x-E-i\eps)^{-1}\to \id_{\{E\}}(x)$ as $\eps\to 0$, the spectral theorem implies  the strong convergence  $-i\eps R_{E+i\eps}(H^{\tau}) \to \id_{\{E\}}(H^{\tau})$ as $\eps\to 0$. Using this, we obtain for some $0<s<1$ and any 
$a, b\in\Z^d$ that
\begin{equation}
\mathbb{E} \big[ \|\chi_a\id_{\{E\}}(H^{\tau}) \chi_b \|^s \big] \leq \liminf_{\eps \to 0} \eps^s\,  \mathbb{E} \big[ \|\chi_a R_{E+i\eps}(H^{\tau})\chi_b \|^{s} \big] = 0,
\end{equation}
where the last equality follows from the a priori bound \eqref{eq:AppAPriori}.
Since $\Z^d$ is countable, $\id_{\{E\}}(H^{\tau})= 0$ almost surely. 
\end{proof}


\section{Deterministic a priori estimate for Schr\"odinger operators}
\label{app:Schatten}

This appendix contains a local Schatten-class bound for functions of (non-random) 
Schr\"odinger operators. 
Estimates of this type appeared in \cite[Sect.\ B.9]{artSEM1982Sim}, \cite[App.\ A]{artRSO2006AizEtAl2} and \cite{MR2303305}. 
We include the proofs for completeness, following closely \cite{MR2303305}.

We consider the non-random Schr\"odinger operator 
\begin{MyDescription}
\item[(D)] {\quad $H:=-\Delta+U$ with a bounded potential $U\in L^{\infty}(\R^{d})$,} \label{AssD}
\end{MyDescription}
acting in $L^{2}(\R^{d})$.
To formulate the deterministic a priori bound we introduce $E_{0} := \inf_{x\in\Rn} U(x)$ 
so that $H_G\geq E_0$ for all $G\subseteq \R^d$ open.

\begin{lemma}
	\label{lem:aPriori}
	Assume {\upshape (D)}. Let $p>0$ and $m\in\N$ with $m>d/(2p)$. Then there exists a finite constant $C$ 
	such that for all open $G\subseteq \R^d$, all $a,b\in\mathbb{R}^{d}$ and all $g\in L^\infty(\R)$
	we have the estimate
	\begin{equation}
		\label{eq:aPrioriStat}
		\|\chi_a g(H_{G})\chi_b\|_p \leq C \|g\|_\infty \big( \max\{0, \sup\supp(g) -E_{0}+1\}  \big)^m   .
	\end{equation}
\end{lemma}

\noindent
The constant $C$ in \eqref{eq:aPrioriStat} depends on the potential $U$, but only through 
$\inf\sigma(H) -E_{0}$.

\begin{proof}
	We use the abbreviation $H:=H_G$. Without loss of generality we assume $E_{0} \le \sup\supp(g) < \infty$ 
	(because otherwise the statement is trivial) and $0<p \le 1$ (because $\|\pmb\cdot\|_{p} \le \|\pmb\cdot\|_{1}$ 
	for $p \ge 1$). Let $m\in\N$ such that $m> d/(2p)$. We insert the $m$-th power of the resolvent 
	and apply H\"older's inequality
	\begin{align}
		\label{eq:aPriori1}
		\big\|\chi_a g(H)\chi_b\big\|_p^{p} &= \| \chi_a(H-E_0+ 1)^{-m}(H-E_0+1)^m g(H)\chi_b\|_p^p \notag \\
		& \leq \| \chi_a(H-E_0+ 1)^{-m}\|_p^p \, \| (H-E_0+1)^m g(H)\chi_b\|^p \notag \\
		& \leq (\sup\supp(g)-E_0+1)^{mp} \, \|g\|_{\infty}^p \, \| \chi_a(H-E_0+ 1)^{-m}\|_p^p.
	\end{align}
	Next we set $a_0:=a$ and estimate with the adapted triangle inequality \eqref{lem:conkav}
	\begin{align}
		\label{eq:aPriori2}
	  \|\chi_a(H  -E_0+1)^{-m}\|_p^p  
		& =  \Big\| \chi_{a_0}(H-E_0+1)^{-1}\Big(\sum_{a_1\in\Zd}\chi_{a_1} \Big)(H-E_0+1)^{-1}  \notag\\
		& \hspace{.7cm} \times \cdots \Big(\sum_{a_{m-1}\in\Zd} \chi_{a_{m-1}}\Big) (H-E_0+1)^{-1} 
			\Big(\sum_{a_m\in\Zd} \chi_{a_m}\Big) \Big\|_p^p \notag \\
		& \leq \sum_{a_1,...,a_m\in \Zd}  \Big\| \prod_{l=1}^{m} \big(\chi_{a_{l-1}}(H-E_0+1)^{-1}\chi_{a_l}\big) \Big\|_p^p.
	\end{align}
	H\"older's inequality for Schatten classes and the Combes-Thomas estimate for Schatten norms 
	in the form of \cite[Lemma A.2]{DGHKM2016} (which equally holds for $H_{G}$) with $pm>d/2$ and $E=E_{0}-1$ imply 
	\begin{align}
		\label{eq:aPriori20}
 		\Big\| \prod_{l=1}^{m} (\chi_{a_{l-1}}(H-E_0+1)^{-1}\chi_{a_l}) \Big\|_p^p
		& \leq \prod_{l=1}^{m} \big\Vert \chi_{a_{l-1}}(H-E_0+1)^{-1}\chi_{a_l} \big\Vert_{pm}^{p} \notag \\
	  & \leq  C_{1}  \prod_{l=1}^{m} \e^{-\mu_{1}|a_{l-1}-a_l|}
	\end{align}
	with finite constants $\mu_{1}=\mu_{1,p,m} >0$ and $C_{1}=C_{1,p,m} >0$. The latter also depends on the 
	potential $U$, but only through $\inf\sigma(H) -E_{0}$.
	Inserting \eqref{eq:aPriori20} into \eqref{eq:aPriori2} and repeatedly using
	\begin{equation}
		\sum_{a_1\in\Zd} \e^{-\mu_{1}|a_0-a_1|} \e^{-\mu_{1}|a_1-a_2|} \leq C_{2} \e^{-\mu_{1}/2 |a_0-a_2|},
	\end{equation}
	where $C_{2}=C_{2,p,m} >0$ is a finite constant, 
	we obtain the lemma from \eqref{eq:aPriori1}.
\end{proof}

\end{appendix}


\newcommand{\noopsort}[1]{} \newcommand{\singleletter}[1]{#1}

\end{document}